\documentclass[AMA,STIX1COL]{WileyNJD-v2}
\usepackage{amssymb}
\usepackage{amsthm}
\usepackage{tabularx}
\usepackage{amsmath}

\articletype{LEARNING-BASED ADAPTIVE CONTROL : THEORY AND APPLICATIONS}%

\begin{document}

\title{Observer-based Adaptive Optimal Output Containment Control problem of Linear Heterogeneous Multi-agent Systems with Relative Output Measurements}

\author[1]{Majid Mazouchi}

\author[1]{Mohammad Bagher Naghibi-Sistani}

\author[1]{Seyed Kamal Hosseini Sani}

\author[2]{Farzaneh Tatari}

\author[3]{Hamidreza Modares}


\address[1]{Department of Electrical Engineering, Ferdowsi University of Mashhad, Mashhad, Iran}

\address[2]{Department of Electrical Engineering,  University of Semnan, Semnan, Iran}

\address[3]{Missouri University of Science and Technology, Rolla, MO 65401, USA}



\abstract[Summary]{This paper develops an optimal relative output-feedback based solution to the containment control problem of linear heterogeneous multi-agent systems. A distributed optimal control protocol is presented for the followers to not only assure that their outputs fall into the convex hull of the leaders' output (i.e., the desired or safe region), but also optimizes their transient performance. The proposed optimal control solution is composed of a feedback part, depending of the followers' state, and a feed-forward part, depending on the convex hull of the leaders' state. To comply with most real-world applications, the feedback and feed-forward states are assumed to be unavailable and are estimated using two distributed observers. That is, since the followers cannot directly sense their absolute states, a distributed observer is designed that uses only relative output measurements with respect to their neighbors (measured for example by using range sensors in robotic) and the information which is broadcasted by their neighbors to estimate their states. Moreover, another adaptive distributed observer is designed that uses exchange of information between followers over a communication network to estimate the convex hull of the leaders' state. The proposed observer relaxes the restrictive requirement of knowing the complete knowledge of the leaders' dynamics by all followers. An off-policy reinforcement learning algorithm on an actor-critic structure is next developed to solve the optimal containment control problem online, using relative output measurements and without requirement of knowing the leaders' dynamics by all followers. Finally, the theoretical results are verified by numerical simulations.}



\maketitle


\section{Introduction}\label{sec1}

Distributed control of multi-agent systems has attracted  a surge of interest in variety of disciplines, due to its broad applications, including cooperation of multiple robot systems \cite{Jad2003,Ji2008}, satellite formation flying \cite{Car2002}, vehicles formation control \cite{Fax2004}, transportation systems \cite{Tom1998}, cooperative surveillance \cite{Ren2007}, distributed sensor networks \cite{Olfati-Saber2005} and so forth. Distributed cooperative control offers many advantages such as less communication requirement, more flexibility, enhanced reliability, and scalability, compared with its centralized counterpart. A fundamental problem in distributed cooperative control of multi-agent systems is consensus or synchronization, in which the goal is to design distributed control policies for agents to ensure that they reach an agreement on certain quantities of interest on their states or outputs, using only the local state or output information available to each agent. A comprehensive review of consensus and synchronization problems is provided in Olfati-Saber et al.\cite{Olfati-Saber2007}. Based on the number of leaders, consensus or synchronization problems can be categorized into three classes, namely, leaderless \cite{Ren2007,Ren2005}, leader-following \cite{Hong2006,Hong2008} and containment control \cite{Ji2008,Lou2012,Notarstefano2011}. In the latest problem, which is the problem of interest in this paper, there exist multiple leaders, and the objective is to drive the followers into a convex geometric space spanned by the leaders. The containment control problem has been extensively investigated in recent years \cite{Ji2008,Lou2012,Notarstefano2011,Meng2010,Cao2012a,Liu2012a,Liu2012b,Mei2012,Li2012,Yoo2013,Li2013,zheng2014,Wang2014,Liu2015,Haghshenas2015,Kan2015,Kan2016}, due to its numerous potential applications in practical engineering, for example, in stellar observation for satellite formation \cite{Dimarogonas2009}, removing hazardous materials for autonomous robots \cite{Ji2008}, and so forth.

In most practical situations, the full state information of agents is unavailable for measurement and/or expensive to measure. For instance, for a group of mobile agents navigating in environments that global navigation satellite systems signals are rather attenuated, such as forests, urban canyons, and even some building interiors, no position measurement might be possible using ordinary global positioning system (GPS) receivers. In this situation, one obvious solution might be attained by installing more precise and powerful GPS receivers on all the agents.  However, not only a more precise and powerful GPS receiver is costly, but also it uses more electrical power due to the fact that it requires more amplifying of attenuated signal and possibly more burden of weight. Therefore, in some real world scenarios, this solution may be unfeasible or too costly. Motivating by the concept of anchor agents \cite{Mao2006} in localization problem \cite{Shames2009} in the context of wireless sensor networks, another solution to the aforementioned situation and scenario is to equip only small fraction of agents (leaders) with more precise and powerful GPS receivers to measure their absolute states. However, the rest of agents (followers), which are just equipped with ordinary GPS receivers, do not have access to their absolute position measurements, and they just have access to relative output measurements with respect to their neighbors and the information which is broadcasted to them through the communication network by their neighbors.

Most of existing containment control protocols focus on the case of homogeneous agents, in which all agents have identical dynamics \cite{Cao2012a,Liu2012a,Li2012,Liu2015,Wen2016}. Some recent results on the containment control problem \cite{zheng2014,Haghshenas2015} have considered the case of heterogeneous followers with non-identical dynamics, but assumed that the dimensions of all agents are the same. However, in many real-world applications, for which there are different types of agents performing different tasks, not only the agents' dynamics but also their dimensions are different. This requires designing distributed control protocols to drive the followers' output into a convex hull spanned by the leaders' output. Nevertheless, existing results based on relative state measurements in zheng et al.\cite{zheng2014} and Haghshenas et al.\cite{Haghshenas2015} cannot be used, as the relative state does not make sense anymore for followers with different dimensions. Although the design of distributed relative output-feedback based control protocols is considered in Li et al.\cite{Li2013} and Wen et al.\cite{Wen2016}, these results are still limited to homogeneous multi-agent systems. Output containment control of heterogeneous multi-agent systems is considered in Zuo et al.\cite{ Zuo2017}. However, in their method, all the followers require their absolute state or output, as well as complete knowledge of the leaders' dynamics, which may not be available to the followers in many applications. Moreover, their approach requires the restrictive assumption of requiring a strongly connected communication graph. 

Besides the above mentioned shortcomings of existing results, another shortcoming is that they do not take into account the transient behavior of the followers and give importance only to the steady-state response of the followers, i.e., they only assure that the followers' states or outputs eventually converge to a convex combination of the leaders' states or outputs. However, it is desired to find optimal solutions that not only guarantee steady-state convergence, but also minimize the transient containment error over time. Another important issue which is not considered in the existing results for containment control is designing online solutions that do not require complete knowledge of the leaders. Reinforcement learning (RL) \cite{Sutton1998,Powell2007} has been successfully used to design adaptive optimal controllers for single-agent systems \cite{Cui2015,Moghadam2017,Vamvoudakis2017,Yasini2014,Modares2012,Zhang2011c,DLiu2016,DLiu2017,DLiu2017b} and multi-agent systems \cite{Modares2016a,Yaghmaie2015,Tatari2016,Tatari2017,Mazouchi2018} online in real time. However, to our knowledge, there is no RL-based solution to the optimal containment control problem.

To overcome the aforementioned shortcomings of the existing work, this paper presents an observer-based adaptive optimal solution to the output containment control problem of linear heterogeneous multi-agent systems, where two distributed observers are used. A distributed adaptive observer is designed to estimate followers' state, and another distributed adaptive observer is developed to estimate the convex hull of the leaders' state. The proposed distributed adaptive observer relaxes the restrictive requirement of knowing the complete knowledge of the leaders' dynamics by all followers. Then, an off-policy reinforcement learning algorithm on an actor-critic structure is developed to solve the optimal output containment control problem online in real time. The proposed algorithm does not require any knowledge of leaders' dynamics and uses only the relative output measured data of the followers and the information which is broadcasted through the communication network by neighbors.
\\
\\
\noindent\textbf{The main contributions of the paper are as follows :}
\begin{enumerate}[1.]
\item A novel distributed dynamic relative output feedback control protocol is developed based on cooperative output regulation framework to solve the output containment control problem of linear fully heterogeneous multi-agent systems.
\item An adaptive distributed observer is presented to estimate the leaders' dynamics, as well as their outputs, and a convex combination of the leaders' states for each follower. In contrast to the existing work, this observer relaxes the restrictive requirement of knowing the complete knowledge of the leaders' dynamics by all followers.
\item An optimal solution to the distributed containment control problem is presented to optimize the transient output containment error of followers as well as their control efforts, while assuring a zero-steady state containment error. 
\item An off-policy RL algorithm is developed to solve the formulated optimal output containment control problem online in real time, using relative output measurements of followers with respect to their neighbors and the information which is broadcasted by neighbors, and without requirement of knowing the complete knowledge of the leaders' dynamics by all followers.
\end{enumerate}

The subsequent sections are organized as follows: some basic concepts of graph theory, definitions and notations are presented in Section 2. Section 3 states the output containment control problem in output regulation framework. Moreover, analysis is provided to find containment control problem offline and non-optimal solution. Distributed adaptive observer is designed in Section 4. The optimality is explicitly imposed in solving the containment control problem in Section 5, which enables us to use RL techniques to learn solution online in real time. Numerical simulation is given to validate the effectiveness of the theoretical results in Section 6. Finally, in Section 7 conclusions are drawn.

\section{Preliminaries}\label{sec2}
\subsection{Notations} \label{subsec2.1}

The following notations are used throughout the paper. Let ${\Re ^n}$ and ${\Re ^{n \times m}}$ represent the $n$ dimensional real vector space and the $n \times m$ real matrix space, respectively. ${0_{m \times n}}$ denotes the $m \times n$ matrix with all zeros. Let ${1_n}$ be a column vector with all entries equal to $1$. ${I_n}$ represents an $n \times n$ identity matrix. $diag\left( {{d_1},...,{d_n}} \right)$ represents a block-diagonal matrix with matrices ${d_1},...,{d_n}$ on its diagonal. ${\left\| . \right\|_2}$ denotes the Euclidean norm. For any matrix ${H_i} \in {\Re ^{n \times q}}$, $i=1,...,m$, $col({H_1},...,{H_m}) = {[H_1^T,...,H_m^T]^T}$ and $Vec({H_i}) = col({H_{i1}},...,{H_{iq}})$ where ${H_{ij}} \in {\Re ^n}$ is the $j$-th column of $H_i$. ${\left[ {H_i} \right]_j}$ denotes the $j$-th row of the matrix ${{H_i}}$. The symbol $ \otimes $ represents the Kronecker product. The distance from $x \in {\Re ^N}$ to the set $\mathcal{C} \subseteq {\Re ^N}$ is denoted by $dist(x,\mathcal{C}) = \mathop {\inf }\limits_{y \in \mathcal{C}} {\left\| {x - y} \right\|_2}$.

\subsection{Graph Theory} \label{subsec2.2}

In this subsection, some basic concepts on algebraic graph theory are briefly reviewed. 
Let the communication topology among $n + m$ agents be presented by a weighted directed acyclic graph $ \mathcal{G} = (V,\mathcal{E},A)$ with a set of nodes $V = \{ {{\nu _1}, \ldots ,{\nu _{n + m}}} \}$, a set of edges $\mathcal{E}  \subseteq V \times V$, and a weighted adjacency matrix $A = [{a_{ij}}]$ with non-negative adjacency elements ${a_{ij}}$, i.e., an edge $( {{\nu _j},{\nu _i}} ) \in \mathcal{E} $ if and only if ${a_{ij}} > 0$. Node ${\nu _j}$ is called the parent node, node  ${\nu _i}$ is the child node, and ${\nu _j}$ is a neighbor of ${\nu _i}$. We assume that there are no self-connections, i.e. $( {{\nu _i},{\nu _i}} ) \notin \mathcal{E} $. A directed graph is acyclic if graph does not have any directed cycle. The set of node ${\nu _i}$ neighbors is denoted by ${\bar N_i} = \{ {{\nu _j} \in V:({\nu _j},{\nu _i}) \in \mathcal{E} ,j \ne i} \}$. A directed path from node ${\nu _i}$ to node ${\nu _j}$ is a sequence of edges $({\nu _i},{\nu _{{k_1}}}),({\nu _{{k_1}}},{\nu _{{k_2}}}),...,({\nu _{{k_\ell }}},{\nu _j})$ with distinct nodes ${\nu _{{k_m}}}$, $ m = 1,...,\ell $ in a directed graph. A directed graph is strongly connected if there is a directed path between every ordered pair of nodes. A directed graph is said to have a spanning forest if there exists at least one node such that there is a directed path from this node to all the other nodes.  An agent is called leader if it does not receive any information from others, i.e., it has no neighbor. Otherwise, it is called a follower. We assume that agents $1$ to $n$ are followers, and agents $n + 1$ to $n + m$ are leaders. For notational convenience, $\mathcal{F} \buildrel \Delta \over = \left\{ {1,...,n} \right\}$ and $\mathcal{R} \buildrel \Delta \over = \left\{ {n + 1,...,n + m} \right\}$ are used to denote, the set of followers and the set of leaders, respectively. The Laplacian matrix $L = [{\ell _{ij}}] \in {\Re ^{(n + m) \times (n + m)}}$ associated with $A$ is defined as ${\ell _{ii}} = \sum\nolimits_{j \in {N_i}} {{a_{ij}}} $ and ${\ell _{ij}} =  - {a_{ij}}$ where $i \ne j$. The Laplacian matrix $L$ associated with $\mathcal{G}$ can be partitioned as 
\begin{align}
L = \left[ {\begin{array}{*{20}{c}}
{{L_1}}&{{L_2}}\\
{{0_{m \times n}}}&{{0_{m \times m}}}
\end{array}} \right]  \label{eq:1}
\end{align}
where ${L_1} \in {\Re ^{n \times n}}$ and ${L_2} \in {\Re ^{n \times m}}$. Note that since the last $m$ agents are the leaders, the last $m$ rows of $L$ are all equal to zero. 

In the sequel, we assume that the communication graph  satisfies the following assumption.

\begin{assumption}
The directed graph $\mathcal{G}$ is acyclic, and for each follower, there exists at least one leader that has a directed path to it. 
\end{assumption}

\begin{lemma}[Meng et al.\cite{Meng2010}]
Under Assumption 1, ${L_1}$ is invertible, all the eigenvalues of ${L_1}$ have positive real parts, each entry of $ - L_1^{ - 1}{L_2}$ is non-negative, and each row of $ - L_1^{ - 1}{L_2}$ has a sum equal to one.
\end{lemma}

\begin{proposition}[Qin et al.\cite{Qin2013}]

A directed acyclic graph $\mathcal{G}$ can be relabeled such that its Laplacian matrix is lower triangular matrix.

\end{proposition}

\begin{definition} [ Rockafellar\cite{Rockafellar1972}]
 A set $\mathcal{C} \subseteq {\Re ^\mathcal{N}}$ is said to be convex if $(1 - \lambda )x + \lambda y \in \mathcal{C}$, whenever $x,\,y \in \mathcal{C}$ and $0 \le \lambda  \le 1$. The convex hull $Co(X)$ of a finite set of points $X = \left\{ {{x_1},...,{x_n}} \right\}$ is the minimal convex set containing all points in $X$. That is $Co(X) = \left\{ {\sum\nolimits_{i = 1}^n {{\alpha _i}{x_i}} \left| {{x_i} \in X,{\alpha _i} \in \Re ,} \right.{\alpha _i} \ge 0,\sum\nolimits_{i = 1}^n {{\alpha _i} = 1} } \right\}$.
\end{definition}

\section{Output Containment Control: An Offline and Non-Opimal Solution}\label{sec3}

In this section, the output containment control problem is first introduced, and some standard assumptions are listed. Then, a distributed dynamic output feedback control protocol is introduced for each follower that uses only the relative output measured data of the followers and the information which is broadcasted 
through the communication network by neighbors. The output containment control problem is then formulated into a linear cooperative output regulation problem and a non-optimal offline solution is provided.

Consider a set of agents with $n$ heterogeneous followers whose models are described as
\begin{align}
&{\dot x_i} = {A_i}{x_i} + {B_i}{u_i}, \label{eq:2} \\
&{y_i} = {C_i}{x_i} \label{eq:3}
\end{align} 
and a set of $m$ homogeneous leaders as
\begin{align}
&{\dot \omega _k} = S{\omega _k}, \label{eq:4} \\
&{y_k} = D{\omega _k} \label{eq:5}
\end{align} 
where $i \in \mathcal{F}$ is the set of followers, $k \in \mathcal{R}$ is the set of leaders, ${x_i} \in {\Re ^{{N_i}}}$ is the state of agent $i$, ${u_i} \in {\Re ^{{p_i}}}$ is its control input, and ${y_i} \in {\Re ^q}$ is its output, ${\omega _k} \in {\Re ^{\bar q}}$ is the state of leader $k$ and ${y_k} \in {\Re ^q}$ is its output.

The following standard assumptions are made on the dynamics of agents\cite{Huang2004,Huang2016}.

\begin{assumption}
The pairs $({A_i},{B_i})$ for $i \in \mathcal{F}$ are stabilizable.
\end{assumption}

\begin{assumption}
The pairs $({A_i},{C_i})$ for $i \in \mathcal{F}$ are detectable, and ${C_i}$ for $i \in \mathcal{F}$, are full row rank.
\end{assumption}

\begin{assumption}
The leaders' dynamics in (\ref{eq:4})-(\ref{eq:5}) are marginally stable, and $D$ is full row rank.
\end{assumption}

\begin{assumption}
The linear matrix equations
  \begin{align}
 &{\Pi _i}S = {A_i}{\Pi _i} + {B_i}{\Gamma _i}, \nonumber \\
 &D = {C_i}{\Pi _i} \label{eq:6}
 \end{align}  	                                                  
have unique solutions ${\Pi _i} \in {\Re ^{{N_i} \times {\bar q}}}$   and ${\Gamma _i} \in {\Re ^{{p_i} \times {\bar q}}}$ for all $i \in \mathcal{F}$.
\end{assumption}

\begin{assumption}
Each follower has access only to the relative output measured data between itself and its neighbors, and the information which is gotten by its neighbors, i.e., the absolute output measurements of followers are not available.
\end{assumption}

Assuming $C_i$ to be full row rank in Assumption 3, is a standard assumption for output feedback (See Gadewadikar et al.\cite{Gadewadikar07}) which is made to avoid redundant measurements. Furthermore, the transmission zeros condition (Huang\cite{Huang2004}, Assumption 1.4) is a standard assumption which is satisfied if $C_i$ is required to be full row rank. Therefore, Assumption 3 is a standard assumption. Regarding Assumption 4, it is worth to mention that, in containment control problem, the leaders' dynamic  cannot  be unstable, due to the fact that convex hull should be bounded, and it is not reasonable to have unbounded convex hull. Assumption 6 is made here to relax the need for the full or relative state measurements in the existing literature, which might not be feasible in some scenarios.

The output containment control problem is now defined as follows.

\begin{definition}
The multi-agent system (\ref{eq:2})-(\ref{eq:5}) achieves the output containment if the followers outputs, converge to the convex hull spanned by the leaders outputs. That is, 
\begin{align}
\mathop {lim}\limits_{t \to \infty } \,dist\left( {{y_i}\left( t \right),Co\left( {{y_k}\left( t \right),k \in \mathcal{R}} \right)} \right) = 0,\,\,\,\,\forall i \in \mathcal{F}. \label{eq:7}
\end{align} 
\end{definition}

Now, define the containment error of the follower $i$ as
\begin{align}
{e_i} = \sum\limits_{j = 1}^n {{a_{ij}}({y_i} - {y_j})}  + \sum\limits_{k = n + 1}^{n + m} {\delta _i^k({y_i} - {y_k})} \label{eq:8}
\end{align} 
which, in compact form, yields
\begin{align}
e = \left( {{L_1} \otimes {I_q}} \right){Y_F} - \sum\limits_{k = n + 1}^{n + m} {\left( {{\Delta ^k}{1_n} \otimes {y_k}} \right)} \label{eq:9}
\end{align} 
where $e$ and ${Y_F}$ are the stack column vectors of ${e_i}$ and ${y_i}$ for $i \in \mathcal{F}$, respectively, ${L_1}$ is defined in (\ref{eq:1}) and ${\Delta ^k}$ denotes the connection weight of leader $k$, which is a diagonal matrix with diagonal elements $\delta _i^k = {a_{ik}}$ where $i \in \mathcal{F}$ and $k \in \mathcal{R}$.

\begin{remark}
${\lim _{t \to \infty }}e(t) = 0$ with $e$ defined in (\ref{eq:9}) implies that the followers' outputs converge to the convex hull spanned by the leaders' outputs
\begin{align}
\sum\limits_{k = n + 1}^{n + m} {\left( {L_1^{ - 1}{\Delta ^k}{1_n} \otimes {y_k}} \right)} \label{eq:10}
\end{align} 
where ${L_1}$ is defined in (\ref{eq:1}). Note that, ${L_2}$ in (\ref{eq:1}) can be expressed as
\begin{align}
{L_2} =  - \left[ {{\Delta ^{n + 1}}{1_n},...,{\Delta ^{n + m}}{1_n}} \right] \label{eq:11}
\end{align} 
then, one has
\begin{align}
\sum\limits_{k = n + 1}^{n + m} {\left( {L_1^{ - 1}{\Delta ^k}{1_n} \otimes {y_k}} \right)}  =  - (L_1^{ - 1}{L_2} \otimes {I_q}){Y_R}. \label{eq:12}
\end{align}
where ${Y_R}$ is the stack column vector of ${y_k}$ for $k \in \mathcal{R}$.
\end{remark}

To solve the containment control problem in Definition 2, the following distributed dynamic measurement relative output feedback control protocol is introduced in this paper.

\begin{align}
{u_i} =& K_i^1{\xi _i} + K_i^2{\eta _i}, \label{eq:13} \\
{{\dot \xi }_i} =& {A_i}{\xi _i} + {B_i}{u_i} - {\mu _i}{F_i}(\sum\limits_{j = 1}^n {{a_{ij}}(({y_j} - {y_i}) + {{\hat y}_i} - {{\hat y}_j})}  + \sum\limits_{k = n + 1}^{n + m} {\delta _i^k(({y_k} - {y_i}) + {{\hat y}_i} - D{\omega _k}} ),  \label{eq:14} \\
{\dot \eta _i} =& S\,{\eta _i} + \beta (\sum\limits_{j = 1}^n {{a_{ij}}({\eta _j} - {\eta _i})}  + \sum\limits_{k = n + 1}^{n + m} {\delta _i^k({\omega _k} - {\eta _i})} ) \label{eq:15} 
\end{align}
for $i \in \mathcal{F}$ , where ${\xi _i} \in {\Re ^{{N_i}}}$ denotes the estimate of the  state  of  agent $i$, ${\hat y_i} = {C_i}{\xi _i}$ denotes the estimate of the output ${y_i}$, ${\eta _i} \in {\Re ^{\bar q}}$  denotes the estimate of the convex combination of the leaders' states. Moreover, ${\mu _i} > 0$ and ${F_i} \in {\Re ^{{N_i} \times q}}$ are coupling gains and gain matrices, respectively, to be designed in Lemma 3, and $\beta  > 0$ is a coupling gain. Finally, $K_i^1 \in {\Re ^{{p_i} \times {N_i}}}$ and $K_i^2 \in {\Re ^{{p_i} \times {\bar q}}}$ are design feedback and feed-forward gain matrices, respectively, for $i \in \mathcal{F}$.

Note that the (\ref{eq:15}) is a dynamic compensator which its compact form can be regarded as an observer for the convex combination of the leaders' states\cite{Haghshenas2015}, and (\ref{eq:14}) is a distributed observer for the state of $i$th follower, i.e., ${x_i}$. Note also that absolute output measurements of followers are not used in this control protocol. 

\begin{remark}
 Since the leaders act as command generators and are usually equipped with more powerful sensors, we reasonably assume that they know their own states and can broadcast them to their neighbors. This assumption is a standard assumption in cooperative control with relative measurements literature (For instant see (Zhang et al.\cite{Zhang2011}, Section IV) and (Wu et al.\cite{Wu2016}, Eqn. (6)). However, the followers are assumed to have not access to their absolute output measurements.
\end{remark}

\begin{remark} Note that the leaders are autonomous agents in the sense that they do not receive information from other agents and acts autonomously to guide other followers. That is, leaders are not being influenced by followers and, therefore, the observer design in (\ref{eq:14}) and as well as the control design in (\ref{eq:13}) are only designed for the followers to assure that they converge into a convex hull of the leaders' output trajectory, which is the ultimate goal of the output containment control problem.
\end{remark}

Consider the distributed state observer (\ref{eq:14}) and let the global state estimation error be
\begin{align}
\tilde X = X - \xi  \label{eq:16}
\end{align} 
where $X$ and $\xi $ are the stack column vectors of ${x_i}$ and ${\xi _i}$ , for $i \in \mathcal{F}$, respectively. Using (\ref{eq:14}) and (\ref{eq:2}), the dynamics of the global state estimation error yields
\begin{align}
\dot {\tilde X} = \left( {A - \mu F\left( {{L_1} \otimes {I_q}} \right)\bar C} \right)\tilde X  \label{eq:17}
\end{align} 
where $A = diag({A_1},...,{A_n})$, $\bar C = diag({C_1},...,{C_n})$, $\mu  = diag\allowbreak({I_{{N_1}}} \otimes {\mu _1},...,{I_{{N_n}}} \otimes {\mu _n})$, and $F = diag({F_1},...,{F_n})$. 

Before proceeding, the following technical results are required.

\begin{lemma}
Let  Assumptions 1 and 3 be satisfied. Let ${\lambda _i}$ be the $i$-th eigenvalue of ${L_1}$. Design the observer gain ${F_i}$ in (\ref{eq:14}) for $i \in \mathcal{F}$ as 
\begin{align}
{F_i} = {\Phi _i}C_i^TR_i^{ - 1}  \label{eq:18}
\end{align} 
where ${\Phi _i}$ is the unique positive definite solution of the observer ARE
\begin{align}
{A_i}{\Phi _i} + {\Phi _i}A_i^T + {E_i} - {\Phi _i}C_i^TR_i^{ - 1}{C_i}{\Phi _i} = 0 \label{eq:19}
\end{align} 
with ${E_i} = E_i^T \in {\Re ^{{N_i} \times {N_i}}}$ and ${R_i} = R_i^T \in {\Re ^{q \times q}}$  positive definite design matrices. Then, the global state estimation error dynamic (\ref{eq:17}) is asymptotically stable if the coupling gain satisfies
\begin{align}
{\mu _i} \ge \frac{1}{{2{\lambda _{i\min }}}} \label{eq:20}
\end{align} 
where ${\lambda _{i\min }} = {\min _{i \in {\bar N_i}}}{\mathop{\rm Re}\nolimits} ({\lambda _i})$.
\end{lemma}

\begin{proof}
One can show that, using Proposition 1 and the same procedure as  Theorem 3 in Fax et al.\cite{Fax2004}, if all the matrices ${A_i} - {\mu _i}{\lambda _i}{F_i}{C_i}$, $\forall i \in \mathcal{F}$ are Hurwitz, then, the global state estimation error dynamic (\ref{eq:17}) is asymptotically stable. Under Assumption 1 and based on Lemma 1, all the eigenvalues of  ${L_1}$ have positive real parts. The rest of proof is similar to that of Theorem 1 in Zhang et al.\cite{Zhang2011} and is omitted.
\end{proof}

\begin{remark} It is worth noting that, the requirement for acycle assumption in Lemma 2 can be obviated if all agents have identical dynamics (See Zhang et al.\cite{Zhang2011} for more details). 
\end{remark}

The composition of (\ref{eq:2})-(\ref{eq:5}), the containment error (\ref{eq:9}), and the control laws (\ref{eq:13})-(\ref{eq:15}) result in the following closed-loop system
\begin{align}
&{{\dot x}_c} = {A_c}{x_c} + {B_c}\Omega \label{eq:21} \\
&\dot \Omega  = \bar S\Omega  \label{eq:22} \\
& e = {C_c}{x_c} + {D_c}\Omega \label{eq:23}
\end{align} 
where
\begin{align}
{{A_c} = \left[ {\begin{array}{*{20}{c}}
A&{\bar B{K_1}}&{\bar B{K_2}}\\
{{A_{c21}}}&{{A_{c22}}}&{\bar B{K_2}}\\
{{0_{n\bar q \times N}}}&{{0_{n\bar q \times N}}}&{{A_{c33}}}
\end{array}} \right],{B_c} = \left[ {\begin{array}{*{20}{c}}
{{0_{N \times m\bar q}}}\\
{{0_{N \times m\bar q}}}\\
{ - ({L_2} \otimes \beta {I_{\bar q}})}
\end{array}} \right],{C_c} = \left[ {\begin{array}{*{20}{c}}
{({L_1} \otimes {I_q})\bar C}&{{0_{nq \times N}}}&{{0_{n\bar q \times n\bar q}}}
\end{array}} \right],{D_c} = {L_2} \otimes D}, \label{eq:24}
\end{align}
${x_c} = {\left[ {\begin{array}{*{20}{c}}
{{X^T}}&{{\xi ^T}}&{{\eta ^T}}
\end{array}} \right]^T}$ is the closed-loop state, $\eta $ is the stack column vectors of ${\eta _i}$, $\Omega $ is the stack column vectors of ${\omega _k}$, ${A_{c21}}={\mu F({L_1} \otimes {I_q})\bar C}$, ${A_{c22}}={A + \bar B{K_1} - \mu F({L_1} \otimes {I_q})\bar C}$, ${A_{c33}}={\bar {\bar S} - ({L_1} \otimes \beta {I_{\bar q}})}$ , $\bar S = {I_m} \otimes S$, $\bar {\bar S} = {I_n} \otimes S$, $A = diag({A_1},...,{A_n})$, $\bar B = diag({B_1},\allowbreak...,{B_n})$, $\bar C = diag({C_1},...,{C_n})$, $N = \sum\limits_{i = 1}^n {{N_i}}$, ${K_1} = diag(K_1^1,...,K_n^1)$ and ${K_2} = diag(K_1^2,...,K_n^2)$.

Now, by using (\ref{eq:21})-(\ref{eq:24}), we are ready to describe the output containment control problem, defined in Definition 2, as a cooperative output regulation problem as follows.

\begin{problem} [Cooperative output regulation problem]
Given the multi-agent system (\ref{eq:2})-(\ref{eq:5}) and digraph $\mathcal{G}$, find the control protocol of the form (\ref{eq:13}) such that the closed-loop system (\ref{eq:21})-(\ref{eq:23})  has the following properties:

\begin{description}
\item[Property 1] The origin of the closed-loop system (\ref{eq:21})-(\ref{eq:23}) with $\Omega $ being set to zero is asymptotically stable, i.e., the matrix ${A_c}$ in (\ref{eq:24}) is Hurwitz. 
\item[Property 2]  For any initial conditions ${x_i}(0)$, ${\xi _i}(0)$, ${\eta _i}(0)$, $i \in \mathcal{F}$ and ${\omega _k}(0)$, $k \in \mathcal{R}$, $\mathop {\lim }\limits_{t \to \infty } e\left( t \right) = 0$.
\end{description}

\end{problem}

\begin{remark}
The Property 1 indicates that all solutions of the closed-loop system (\ref{eq:21})-(\ref{eq:23}) forget their initial conditions and converge to zero when the exosystems (i.e., leaders) are disconnected. Moreover, Properties 1 and 2 together indicate that all solutions of the closed-loop system (\ref{eq:21})-(\ref{eq:23}) forget their initial conditions and converge to some solutions, depending only on the exosignals (i.e., ${\omega _k}$ for $k \in \mathcal{R}$). (See (Haghshenas et al.\cite{Haghshenas2015}, Problem 3) and (Huang\cite{Huang2016}, Problem 1 and 2) for more details)
\end{remark}

\begin{remark} Note that, the error given in (\ref{eq:23}) is actually the output containment error of (\ref{eq:9}). Therefore, based on Remark 1, by solving Problem 1, the output containment control problem, described in Definition 2, is also solved. 
\end{remark}

In order to solve Problem 1, we present following lemma on the closed-loop system (\ref{eq:21})-(\ref{eq:23}).

\begin{lemma}
Suppose that the closed-loop system (\ref{eq:21})-(\ref{eq:23}) satisfies Property 1, under the distributed control laws (\ref{eq:13})-(\ref{eq:15}). Then, $\mathop {\lim }\limits_{t \to \infty } e(t) = 0$, if there exists a matrix ${X_c}$ that satisfies the following linear matrix equations:
\begin{align}
\left\{ \begin{array}{l}
{A_c}{X_c} + B_c = {X_c}\bar S\\
{C_c}{X_c} + D_c = {0_{nq \times m\bar q}}
\end{array} \right.  \label{eq:25}
\end{align} 
\end{lemma}

\begin{proof}
Let ${\tilde x_c} = {x_c} - {X_c} \Omega $. Using (\ref{eq:25}), one has 

\begin{align}
\begin{array}{l}
{{\dot{\tilde x}}_c} = {A_c}{{\tilde x}_c}
\end{array} \label{eq:26}
\end{align}

Since ${A_c}$ is Hurwitz, $\mathop {\lim }\limits_{t \to \infty } {\tilde x_c}(t) = 0$. Using (\ref{eq:25}) and some manipulation, (\ref{eq:23}) can be rewritten as
\begin{align}
\begin{array}{l}
e = {C_c}{{\tilde x}_c}
\end{array}  \label{eq:27}
\end{align}
then, it follows from $\mathop {\lim }\limits_{t \to \infty } {\tilde x_c}(t) = 0$, that
$\mathop {\lim }\limits_{t \to \infty } e(t) = 0$.
\end{proof}

Now the following theorem shows that Problem 1 can be solved using the distributed control laws (\ref{eq:13})-(\ref{eq:15}).

\begin{theorem} \label{thm1}
Consider the multi-agent system (\ref{eq:2})-(\ref{eq:5}). Let Assumptions 1 and 5 be satisfied and ${\mu _i} $  and ${F_i}$ in the control protocol (\ref{eq:13}) be designed as of Lemma 2.  Design $K_i^1$ such that ${A_i} + {B_i}K_i^1$ is Hurwitz and $K_i^2$ using
\begin{align}
K_i^2 = {\Gamma _i} - K_i^1{\Pi _i}.  \label{eq:28}
\end{align} 
where ${\Gamma _i}$ and ${\Pi _i}$ are solutions of (\ref{eq:6}). Then, Problem 1 is solved using the distributed control laws (\ref{eq:13})-(\ref{eq:15}), for any positive constant $\beta $.
\end{theorem}

\begin{proof}
For nonsingular matrix 
\begin{align}
T = \left[ {\begin{array}{*{20}{c}}
{{I_N}}&{{0_{N \times N}}}&{{0_{N \times n{\bar q}}}}\\
{ {I_N}}&{{I_N}}&{{0_{N \times n{\bar q}}}}\\
{{0_{n{\bar q} \times N}}}&{{0_{n{\bar q} \times N}}}&{{I_{n{\bar q}}}}
\end{array}} \right], \nonumber
\end{align}
one can verify that
\begin{align}
\begin{array}{l}
{T^{ - 1}}{A_c}T = \left[ {\begin{array}{*{20}{c}}
{A + \bar B{K_1}}&{\bar B{K_1}}&{\bar B{K_2}}\\
{{0_{N \times N}}}&{A - \mu F({L_1} \otimes {I_q})\bar C}&{{0_{N \times n\bar q}}}\\
{{0_{n\bar q \times N}}}&{{0_{n\bar q \times N}}}&{\bar {\bar S} - ({L_1} \otimes \beta {I_{\bar q}})}
\end{array}} \right] \label{eq:29}
\end{array}
\end{align}
which has the block-triangular structure. Under Assumption 2, there exists $K_i^1$, $\forall i \in \mathcal{F}$ such that ${A_i} + {B_i}K_i^1$ is Hurwitz. Based on Lemma 2, $A - \mu F({L_1} \otimes {I_q})\bar C$ is Hurwitz, if positive constants ${\mu _i}$ and ${F_i}$, $\forall i \in \mathcal{F}$ are chosen as (\ref{eq:20}) and (\ref{eq:18}), respectively. Based on Lemma 1, under Assumption 1, all the eigenvalues of ${L_1}$ have positive real parts. Therefore, under Assumption 4, one can see that $\bar {{\bar S}} - ({L_1} \otimes \beta {I_{\bar q}})$ is Hurwitz for any positive constant $\beta$. Thus, under the distributed control laws (\ref{eq:13})-(\ref{eq:15}), Property 1 in Problem 1 is satisfied.

Next we verify the Property 2 in Problem 1. Let $K_i^2$, $\forall i \in \mathcal{F}$ be given by (\ref{eq:28}). Then, under Assumption 5, one has
\begin{align}
(A + \bar B{K_1})\Pi  + \bar B{K_2} = \Pi \left( {{I_n} \otimes S} \right) \label{eq:30}
\end{align} 
where $\Pi  = diag({\Pi _i})$ for $i \in \mathcal{F}$, Let
\begin{align}
{X_c}\buildrel \Delta \over = \left[ {\begin{array}{*{20}{c}}
{ - {{\bar C}^{ - 1}}(L_1^{ - 1}{L_2} \otimes D)}\\
{ - {{\bar C}^{ - 1}}(L_1^{ - 1}{L_2} \otimes D)}\\
{ - (L_1^{ - 1}{L_2} \otimes {I_{\bar q}})}
\end{array}} \right]. \label{eq:31}
\end{align}

Then, using (\ref{eq:21})-(\ref{eq:23}) and (\ref{eq:30}),  yields
\begin{align}
\begin{array}{l}
{A_c}{X_c} + {B_c} = - \left[ {\begin{array}{*{20}{c}}
{(A + \bar B{K_1}){{\bar C}^{ - 1}}(L_1^{ - 1}{L_2} \otimes D) + \bar B{K_2}(L_1^{ - 1}{L_2} \otimes {I_{\bar q}})}\\
{(A + \bar B{K_1}){{\bar C}^{ - 1}}(L_1^{ - 1}{L_2} \otimes D) + \bar B{K_2}(L_1^{ - 1}{L_2} \otimes {I_{\bar q}})}\\
{(\bar {\bar S} - ({L_1} \otimes \beta {I_{\bar q}}))(L_1^{ - 1}{L_2} \otimes {I_{\bar q}}) + ({L_2} \otimes \beta {I_{\bar q}})}
\end{array}} \right]\\
 \quad \quad \quad \quad \,\,\,\,=  - \left[ {\begin{array}{*{20}{c}}
{{{\bar C}^{ - 1}}(L_1^{ - 1}{L_2} \otimes D)\bar S}\\
{{{\bar C}^{ - 1}}(L_1^{ - 1}{L_2} \otimes D)\bar S}\\
{(L_1^{ - 1}{L_2} \otimes {I_{\bar q}})\bar S}
\end{array}} \right] = {X_c}\bar S. \label{eq:32}
\end{array}
\end{align}

Note that
\begin{align}
\bar {\bar S}(L_1^{ - 1}{L_2} \otimes {I_{\bar q}}) &= ({I_n} \otimes S)(L_1^{ - 1}{L_2} \otimes {I_{\bar q}})\nonumber \\ 
 &= (L_1^{ - 1}{L_2} \otimes {I_{\bar q}})\bar S.  \label{eq:33}
\end{align} 

Moreover, 
\begin{align}
\begin{array}{l}
{C_c}{X_c} + {D_c} =  - ({L_1} \otimes {I_q})\bar C{{\bar C}^{ - 1}}(L_1^{ - 1}{L_2} \otimes D) + ({L_2} \otimes D)\\
 \qquad \qquad \,\,\,\,\,= 0.
\end{array} \label{eq:34}
\end{align} 

Hence, ${X_c}$ satisfies the linear matrix equations (\ref{eq:25}) and it follows from Lemma 3 that property 2 is also satisfied, i.e., $\mathop {\lim }\limits_{t \to \infty } \,\,e(t) = 0$. This completes the proof.
\end{proof}

\begin{remark}
 It is seen from (\ref{eq:31}) in the proof of Theorem 1 that, Assumption 1-5 guarantees the feasibility of the solution in Lemma 3, i.e., ${X_c}$.
\end{remark}

\begin{remark} 
The graph condition in Haghshenas et al.\cite{Haghshenas2015} and Zuo et al.\cite{Zuo2017} is unnecessarily strong. In this paper, by defining the output containment error as (\ref{eq:9}), the restrictive required assumption of strongly connected communication graph in (Haghshenas et al.\cite{Haghshenas2015}, Assumption 6) and (Zuo et al.\cite{Zuo2017}, Assumption 1) is relaxed to a milder assumption of communication graph with spanning forest. 
\end{remark}

\begin{remark} 
 The solutions to the output regulator equations (\ref{eq:6}) as well as the distributed observer (\ref{eq:15}) need complete knowledge of leaders' dynamics. This knowledge, however, is not available to the followers in many applications. In order to obviate the requirement of the leaders' dynamics in (\ref{eq:15}), a distributed adaptive observer is designed for convex combination of the leaders' states in Section 4.  Optimality is next implicitly incorporated in the design of the containment control problem in Section 5 to optimize the transient containment errors of agents. An off-policy RL algorithm is developed to solve the optimal output containment control problem online in real time and without requiring the knowledge of the leaders' dynamics.
\end{remark}

\section{Distributed Adaptive Observer for Leaders Convex Hull}\label{sec4}

In this section, a distributed adaptive observer is developed to estimate convex combination of the leaders' states and outputs for each follower, as well as the leaders' dynamics, simultaneously.
To estimate the convex combination of the leaders' states and outputs for each follower, consider the following distributed adaptive observer 
\begin{align}
\dot {\hat {\eta _i}} = {S_i}{\hat \eta _i} + {\beta _2}(\sum\limits_{j = 1}^n {{a_{ij}}({{\hat \eta }_j} - {{\hat \eta }_i})}  + \sum\limits_{k = n + 1}^{n + m} {\delta _i^k({\omega _k} - {{\hat \eta }_i})} )  \label{eq:35}
\end{align}
where
\begin{align}
&{\dot S_i} = {\beta _1}(\sum\limits_{j = 1}^n {{a_{ij}}({S_j} - {S_i})}  + \sum\limits_{k = n + 1}^{n + m} {\delta _i^k(S - {S_i})} )  \label{eq:36} \\
&{\dot D_i} = {\beta _3}(\sum\limits_{j = 1}^n {{a_{ij}}({D_j} - {D_i})}  + \sum\limits_{k = n + 1}^{n + m} {\delta _i^k(D - {D_i})})  \label{eq:37} \\
&{y_{0i}} = {D_i}{\hat \eta _i} \label{eq:38}
\end{align} 
where ${S_i} \in {\Re ^{{\bar q} \times {\bar q}}}$ and ${D_i} \in {\Re ^{q \times {\bar q}}}$ , for $i \in \mathcal{F}$, are the estimation of the leaders' dynamics $S$ and $D$ respectively, ${\hat \eta _i} \in {\Re ^{\bar q}}$ and ${y_{0i}} \in {\Re ^q}$ are the estimation of convex combination of the leaders' states and outputs for $i$th follower, respectively, and ${\beta _1},{\beta _2}, {\beta _3} > 0$.  The following lemma is used in the proof of Theorem 2 for the observer design. 

\begin{lemma} [Cai et al.\cite{Cai2017}]
Consider the following system
\begin{align}
\dot \upsilon  = \varepsilon {\rm Z}\upsilon  + {{\rm Z}_1}(t)\upsilon  + {{\rm Z}_2}(t) \label{eq:39}
\end{align} 
where $\upsilon  \in {\Re ^{\bar q}}$, ${\rm Z} \in {\Re ^{{\bar q} \times {\bar q}}}$ is asymptotically stable, $\varepsilon $ is a positive constant, ${{\rm Z}_1}(t) \in {\Re ^{{\bar q} \times {\bar q}}}$ and ${{\rm Z}_2}(t) \in {\Re ^{\bar q}}$ are bounded and continuous for all $t \ge {t_o}$. If ${{\rm Z}_1}(t)$ and ${{\rm Z}_2}(t)$ decay to zero exponentially as time go to infinity, then, for any $\upsilon ({t_0})$ and any $\varepsilon  > 0$, $\upsilon (t)$ decays to zero exponentially as time go to infinity.
\end{lemma}

\begin{theorem} \label{thm2}
Consider the leader dynamics (\ref{eq:4})-(\ref{eq:5}), and the adaptive observer (\ref{eq:35})-(\ref{eq:37}). Let ${\tilde S_i} = {S_i} - S$,  ${\tilde D_i} = {D_i} - D$ be leaders' dynamics estimation errors and ${\tilde \eta _i} = {\hat \eta _i} - \omega _i^*$ and ${\tilde y_{0i}}(t) = {D_i}{\hat \eta _i} - D\omega _i^*$ be $i$th follower state estimation error of the convex combination of the leaders' states and outputs, respectively, where $\omega _i^*$ is $ - ({\left[ {L_1^{ - 1}{L_2}} \right]_i} \otimes {I_{\bar q}})\omega$ with $\omega  = col({\omega _1},...,{\omega _m})$. Then, for any initial conditions ${\tilde S_i}(0)$, ${\tilde D_i}(0)$, ${\tilde y_{0i}}(0)$ and ${\tilde \eta _i}(0)$, one obtains

\begin{enumerate}[1.]
\item for any positive constant ${\beta _1}$, $\forall i \in \mathcal{F}$, ${\lim _{t \to \infty }}{\tilde S_i}(t) = 0$  exponentially; 

\item for any positive constant ${\beta _3}$, $\forall i \in \mathcal{F}$, ${\lim _{t \to \infty }}{\tilde D_i}(t) = 0$ exponentially; 

\item for any positive constant ${\beta _1}$ and  ${\beta _2}$, $\forall i \in \mathcal{F}$, ${\lim _{t \to \infty }}{\tilde \eta _i}(t) = 0$ exponentially;

\item for any positive constant ${\beta _1}$, ${\beta _2}$ and ${\beta _3}$ , $\forall i \in \mathcal{F}$, ${\lim _{t \to \infty }}{\tilde y_{0i}}(t) = 0$ exponentially.
\end{enumerate}
\end{theorem}

\begin{proof}
The dynamics of the global leaders' dynamics estimation error can be written as
\begin{align}
\dot {\tilde {S}} =  - {\beta _1}({L_1} \otimes {I_q})( {{S_c} + (L_1^{ - 1}{L_2} \otimes {I_{\bar q}})({1_m} \otimes S)} ) \label{eq:40}
\end{align} 
where ${S_c} = col({S_1},...,{S_n})$. Using Lemma 1, it is easy to show that
\begin{align}
- (L_1^{ - 1}{L_2} \otimes {I_{\bar q}})({1_m} \otimes S) = ({1_n} \otimes S).  \label{eq:41}
\end{align}
 
Then,
\begin{align}
\dot {\tilde {S}} =  - {\beta _1}({L_1} \otimes {I_{\bar q}})\tilde S
 \label{eq:42}
\end{align} 
or equivalently
\begin{align}
Vec(\dot {\tilde {S}}) =  - {\beta _1}({I_{\bar q}} \otimes {L_1} \otimes {I_{\bar q}})Vec(\tilde S).
 \label{eq:43}
\end{align} 

According to Assumption 1 and Lemma 1, all the eigenvalues of ${L_1}$  have positive real parts. Hence, ${\lim _{t \to \infty }}Vec(\tilde S(t)) = 0$ exponentially, which yields ${\lim _{t \to \infty }}{\tilde S_i}(t) = 0$, $\forall i \in \mathcal{F}$. Similar to part (1), one can see that ${\lim _{t \to \infty }}{\tilde D_i}(t) = 0$, $\forall i \in \mathcal{F}$ exponentially for any positive constant ${\beta _3}$.

Next, it remains to prove part (3) and part (4), i.e. ${\lim _{t \to \infty }}{\tilde \eta _i}(t) = 0$ and ${\lim _{t \to \infty }}{\tilde y_{0i}}(t) = 0$. To this end, using (\ref{eq:35}) and some manipulations yields
\begin{align}
{{\dot {\tilde{ \eta}} }_i} =& {S_i}{{\hat \eta }_i} + {\beta _2}(\sum\limits_{j = 1}^n {{a_{ij}}({{\hat \eta }_j} - {{\hat \eta }_i})}  + \sum\limits_{k = n + 1}^{n + m} {\delta _i^k({\omega _k} - {{\hat \eta }_i})} )+({\left[ {L_1^{ - 1}{L_2}} \right]_i} \otimes {I_{\bar q}})({I_m} \otimes S)\omega   \nonumber \\
 =&{\beta _2}(\sum\limits_{j = 1}^n {{a_{ij}}({{\hat \eta }_j} - {{\hat \eta }_i})}  + \sum\limits_{k = n + 1}^{n + m} {\delta _i^k({\omega _k} - {{\hat \eta }_i})} )+S{{\tilde \eta }_i} + {{\tilde S}_i}{{\tilde \eta }_i} + {{\tilde S}_i}\omega _i^*.  \label{eq:44}
\end{align}

The global form of (\ref{eq:44}) can be written as
\begin{align}
\dot {\tilde {\eta}}  =& (({I_n} \otimes S))\tilde \eta  + {{\tilde S}_d}\tilde \eta  + {{\tilde S}_d}{\omega ^*} - {\beta _2}({L_2} \otimes {I_{\bar q}})\omega- {\beta _2}({L_1} \otimes {I_{\bar q}})(\tilde \eta  - (L_1^{ - 1}{L_2} \otimes {I_{\bar q}})\omega ) \nonumber \\
 =& (({I_n} \otimes S) - {\beta _2}({L_1} \otimes {I_{\bar q}}))\tilde \eta  + {{\tilde S}_d}\tilde \eta  + {{\tilde S}_d}{\omega ^*}
 \label{eq:45}
\end{align} 
where $\tilde \eta  = col({\tilde \eta _1},...,{\tilde \eta _n})$, ${\omega ^*} = col(\omega _1^*,...,\omega _n^*)$ and ${\tilde S_d} = diag\allowbreak({\tilde S_1},...,{\tilde S_n})$. According to Assumption 4, for any ${\beta _1},{\beta _2} > 0$, the matrix $( {({I_n} \otimes S) - {\beta _2}({L_1} \otimes {I_{\bar q}})} )$ is Hurwitz, and ${\tilde S_d}{\omega ^*}$ decays to zero exponentially. Hence, applying Lemma 4 we conclude that ${\lim _{t \to \infty }}{\tilde \eta _i}(t) = 0$  exponentially. Adding and subtracting ${D_i}\omega _i^*$ to the right hand side of ${\tilde y_{0i}}(t)$, we have
\begin{align}
{\tilde y_{0i}}(t) = {D_i}({\hat \eta _i} - \omega _i^*) + ({D_i} - D)\omega _i^*.
 \label{eq:46}
\end{align}

Therefore, if ${\lim _{t \to \infty }}{\tilde \eta _i}(t) = 0$ and ${\lim _{t \to \infty }}{\tilde D_i}(t) = 0$, then ${\lim _{t \to \infty }}{\tilde y_{0i}}(t) = 0$, which completes the proof.
\end{proof}

\section{Optimal Output Containment Control Problem}\label{sec5}

In this section, first an optimal formulation for the output containment control problem is introduced. Then, an off-policy reinforcement learning (RL) algorithm is proposed to drive the followers optimally into a convex hull spanned by the leaders' outputs. At first, it is assumed that followers can have access to their states and the leaders' states and dynamics. Then, this restrictive assumption is relaxed by combining the RL based optimal control presented in this section with distributed adaptive observer designed in Section 4 and the distributed observer (\ref{eq:14}).

\subsection{Problem Formulation and Its Offline Solution}

The aim of optimal output containment control problem is to design the distributed control laws (\ref{eq:13})-(\ref{eq:15}) such that the followers converge into a convex hull spanned by the leaders' outputs, and at the same time, the transient output containment error captured by (\ref{eq:9}) as well as the followers' control efforts  are minimized. To do so, an optimal output containment control problem is defined as follows.

\begin{problem} [Optimal output containment control problem]
Consider the control protocol (\ref{eq:13}) and let the distributed state observer (\ref{eq:14}) be designed as of Lemma 2. Moreover, let the distributed observers (\ref{eq:35})-(\ref{eq:37}) be designed as of Theorem 2. Design the gain matrices $K_i^1$ and $K_i^2$, $\forall i \in \mathcal{F}$ in the control protocol (\ref{eq:13}), so that not only properties 1 and 2 in Problem 1 are satisfied, but also the following discounted performance function is minimized.
\begin{align}
{V_i}\left( t \right) = \int\limits_t^\infty  {{e^{ - {\gamma _i}(\tau  - t)}}({{\bar e}_i}^T{Q_i}{{\bar e}_i} + {u_i}^T{W_i}{u_i})d\tau } 
 \label{eq:47}
\end{align} 
for $i \in \mathcal{F}$,where ${Q_i} \in {\Re ^{q \times q}}$ and ${W_i}= diag(w_i^1,...,w_i^{{p_i}}) \in {\Re ^{{p_i} \times {p_i}}}$ are symmetric positive definite weight matrices, $w_i^j > 0$ for $j=1,...,p_i$, ${\gamma _i} > 0$ is the discount factor, and 
\begin{align}
{\bar e_i} = {y_i} - D\omega _i^*
 \label{eq:48}
\end{align} 
is the containment error with $\omega _i^* =  - ({\left[ {L_1^{ - 1}{L_2}} \right]_i} \otimes {I_{\bar q}})\omega $, $\omega  = col({\omega _1},...,{\omega _m})$.
\end{problem} 

\begin{remark} Note that (\ref{eq:9}) can be rewritten as follows
\begin{align}
e = \left( {{L_1} \otimes {I_q}} \right)\bar e
 \label{eq:49}
\end{align} 
where
\begin{align}
\bar e = {Y_F} + \left( {L_1^{ - 1}{L_2} \otimes {I_q}} \right){Y_R}.
 \label{eq:50}
\end{align}

Therefore, according to Lemma 1, one can see that minimizing $\bar e$ results in minimizing the output containment error (\ref{eq:9}). It was shown that if the state observer gains in (\ref{eq:14}) are designed as of Lemma 2, then ${\xi _i} \to {x_i}$. Moreover, it was shown that if the distributed observers (\ref{eq:35})-(\ref{eq:37}) are designed as of Theorem 2, then ${S_i} \to S$, ${D_i} \to D$, and ${\hat \eta _i} \to \omega _i^*$. Therefore, in steady-state, the control protocol (\ref{eq:13}) becomes 
\begin{align}
{u_i} = K_i^1{x_i} + K_i^2\omega _i^*.
 \label{eq:51}
\end{align} 
\end{remark}

It is shown in the following that if this control protocol is designed to minimize (\ref{eq:47}), then the output containment control Problem 2 is solved. One can write (\ref{eq:51}) in the following form 
\begin{align}
{u_i} = {\bar K_i}{\bar X_i}
 \label{eq:52}
\end{align} 
where ${\bar K_i} = [K_i^1,K_i^2]$ and ${\bar X_i}(t) = {[{x_i}^T, {\omega_i^*}^T]^T}$ denotes an augmented state.

The augmented dynamics is given by 
\begin{align}
&{\dot {\bar X}_i} = {P_i}{\bar X_i} + {\bar B_i}{u_i} \label{eq:53} \\
&{\bar e_i} = {\bar {\bar C}_i}{\bar X_i}  \label{eq:54}
\end{align} 
with
\begin{align}
&{P_i} = \left[ {\begin{array}{*{20}{c}}
{{A_i}}&0\\
0&S
\end{array}} \right],
{\bar B_i} = \left[ {\begin{array}{*{20}{c}}
{{B_i}}\\
0
\end{array}} \right] \label{eq:55} \\
&{\bar {\bar C}_i} = \left[ {\begin{array}{*{20}{c}}
{{C_i}}&{ - D}
\end{array}} \right].
 \label{eq:56}
\end{align} 

Now, the value function  (\ref{eq:47}) can be written as the following quadratic form
\begin{align}
V({X_i},{u_i}) &= \int\limits_t^\infty  {{e^{ - {\gamma _i}(\tau  - t)}}{{\bar X}_i}^T({{\bar{ \bar C}}_i}^T{Q_i}{{\bar {\bar C}}_i} + {{\bar K}_i}^T{W_i}{{\bar K}_i}){{\bar X}_i}d\tau } \nonumber \\
&= {{\bar X}_i}^T{\Psi_i}{{\bar X}_i}.
 \label{eq:57}
\end{align}

Hamiltonian is defined as follows

\begin{align}
H({\bar X_i},{u_i},V({X_i})) = \frac{{dV({X_i})}}{{dt}} - {\gamma _i}V({X_i}) + {\bar X_i}^T({\bar {\bar C}_i}^T{Q_i}{\bar {\bar C}_i} + {\bar K_i}^T{W_i}{\bar K_i}){\bar X_i}  \label{eq:58} 
\end{align}

Then, using stationary condition\cite{LewisBook}, i.e., ${{\partial H} \mathord{\left/
 {\vphantom {{\partial H} {\partial {u_i}}}} \right.
 \kern-\nulldelimiterspace} {\partial {u_i}}} = 0$, the optimal control gain in (\ref{eq:52}) is derived as follows 
\begin{align}
u_i^* = \bar K_i^*{\bar X_i}
 \label{eq:59}
\end{align} 
with
\begin{align}
{\bar K_i}^* &= [ {K{{_i^1}^*},K{{_i^2}^*}} ] \nonumber \\
 &=  - W_i^{- 1}\bar B_i^T{\Psi _i} \label{eq:60}
\end{align} 
where ${\Psi_i}$ satisfies the discounted algebraic Riccati equation (ARE) (\ref{eq:61}).

Substituting (\ref{eq:59}), (\ref{eq:60}), and (\ref{eq:57}) into (\ref{eq:58}), yields the discounted ARE as follows 

\begin{align}
P_i^T{\Psi_i} + {P_i}{\Psi_i} - {\gamma _i}{\Psi_i} + {\bar {\bar C}_i}^T{Q_i}{\bar {\bar C}_i} - {\Psi_i}{\bar B_i}W_i^{ - 1}\bar B_i^T{\Psi_i} = 0.
 \label{eq:61}
\end{align} 

\begin{remark}
 Note that the ARE (\ref{eq:61}) and consequently the control gain (\ref{eq:60}) requires complete knowledge of the leaders' dynamics. To learn the gains (\ref{eq:60}) in an online fashion without requiring complete knowledge of the leaders' dynamics, an off-policy RL algorithm is designed in Subsection 5.2.
\end{remark}

\begin{remark}
 It is shown in Modares et al.\cite{Modares2015} that if discount factor ${\gamma _i}$ in (\ref{eq:47}) satisfies 
\begin{align}
{\gamma _i} \le \gamma _i^* = 2\left\| {{{({{\bar B}_i}W_i^{ - 1}\bar B_i^T{Q_i})}^{{1 \mathord{\left/
 {\vphantom {1 2}} \right.
 \kern-\nulldelimiterspace} 2}}}} \right\|,
 \label{eq:62}
\end{align} 
then, the value function (\ref{eq:47}) is bounded for the control policy (\ref{eq:59}) and the closed-loop agents' dynamics are stable. This is equivalent to Property 1 in Problem 1. The following results show that solving Problem 2 actually solves Problem 1 in an optimal manner.
\end{remark}

\begin{lemma} Consider the multi-agent system (\ref{eq:2})-(\ref{eq:5}) and the control policy (\ref{eq:59}) with gain given by (\ref{eq:60}) and (\ref{eq:61}). If the discount factors ${\gamma _i}$ satisfies (\ref{eq:62}), then $\bar e = {Y_F} + \left( {L_1^{ - 1}{L_2} \otimes {I_q}} \right){Y_R}$ and consequently output containment error (\ref{eq:9}) goes to zero asymptotically.
\end{lemma}

\begin{proof}
 First, let the ARE (\ref{eq:61}) be written in the compact form as follows
\begin{align}
{P^T}\Psi + P\Psi - \bar \gamma \Psi + {\bar {\bar C}^T}\bar Q\bar {\bar C} - \Psi\bar B{W^{ - 1}}{\bar B^T}\Psi =0 
 \label{eq:63}
\end{align} 
where $P = diag({P_1},...,{P_n})$, $\Psi = diag({\Psi_1},...,{\Psi_n})$,  $\bar \gamma  = diag\allowbreak({\gamma _1},...,{\gamma _n})$, $W = diag({W_1},...,{W_n})$, $\bar {\bar C} = diag({\bar {\bar C}_1},...,{\bar {\bar C}_n})$, $\bar Q = diag({Q_1},...{Q_n})$, and $\bar B = diag({B_1},...,{B_n})$.  Multiplying ${\bar X^T}$ and $\bar X$ to right and left sides of (\ref{eq:63}), respectively, one has 
\begin{align}
2{\bar X^T}{P^T}\Psi\bar X - {\bar X^T}\bar \gamma \Psi\bar X + {\bar X^T}{\bar {\bar C}^T}\bar Q\bar {\bar C}\bar X - {\left( {\Psi\bar X} \right)^T}\bar B{W^{ - 1}}{\bar B^T}\Psi\bar X = 0.
 \label{eq:64}
\end{align} 
Observing (\ref{eq:64}), one can see that, the null space of $\Psi$ is a subspace of the null space of ${\bar {\bar C}^T}\bar Q\bar {\bar C}$. Using this observation, it can be seen that if ${\bar X^T}\Psi\bar X = 0$, then ${\bar X^T}{\bar {\bar C}^T}\bar Q\bar {\bar C}\bar X = 0$ and consequently ${({Y_F} - (L_1^{ - 1}{L_2} \otimes {I_q}){Y_R})^T}{\bar Q}({Y_F} - (L_1^{ - 1}{L_2} \otimes {I_q}){Y_R}) = 0$ which yields ${Y_F} + (L_1^{ - 1}{L_2} \otimes {I_q}){Y_R} = 0$ and consequently the containment error $e = 0$, which is defined in (\ref{eq:9}). Therefore, one can conclude that the null space of $\Psi$ is a subspace in which the containment error $e$ is zero. It remains to show that providing the discount factors ${\gamma _i}$ satisfies the (\ref{eq:62}), then the null space of $\Psi$ is attractive. To this end, choose the following Lyapunov function
\begin{align}
V(\bar X) = {\bar X^T}\Psi\bar X.
 \label{eq:65}
\end{align} 
Taking the derivative of (\ref{eq:65}) gives
\begin{align}
\begin{array}{l}
\dot V(\bar X) = \sum\limits_{i \in \mathcal{F}} {{{\bar X}_i}^T\left( {{\Psi _i}{A_{ci}} + A_{ci}^T{\Psi _i}} \right){{\bar X}_i}} \\
 \quad \quad \, \, \, \,= {{\bar X}^T}\left( {\Psi {{\bar A}_c} + \bar A_c^T\Psi } \right)\bar X
\end{array}  \label{eq:66}
\end{align} 
where ${\bar A_c} = diag({A_{c1}},...,{A_{cn}})$ and
\begin{align}
{A_{ci}} = \left[ {\begin{array}{*{20}{c}}
{{A_i} + {B_i}K{{_i^1}^*}}&{{B_i}K{{_i^2}^*}}\\
0&S
\end{array}} \right]. \label{eq:67}
\end{align} 
 
It was mentioned in Remark 12 that ${A_i} + {B_i}K_i^1$ is Hurwitz, providing that the discount factors ${\gamma _i}$ satisfies the (\ref{eq:62}). Therefore, under Assumption 4, one can conclude $\forall i \in {\cal F}$, ${A_{ci}}$ and consequently ${\bar A_c}$ are marginally stable. Therefore, there exist a positive semi definite matrix $\Upsilon  =  - \left( {\Psi {{\bar A}_c} + \bar A_c^T\Psi } \right)$ such that $\dot V(\bar X) =  - {\bar X^T}\Upsilon \bar X \le 0$. Based on LaSalle's invariance principle, $\bar X$ converges to the largest invariance subspace where $\dot V(\bar X) = 0$. As mentioned above, the null space of $\Psi$ is a subspace of where the output containment error $e$ is equal to zero, therefore based on (\ref{eq:66}), $\dot V(\bar X) = 0$ if $\Psi X = 0$ and consequently the output containment error $e$ is equal to zero, which completes the proof.
\end{proof}

Note that ${{ \bar X}_i}(t) = {[{x_i}^T,{\omega_i^*}^T]^T}$ in (\ref{eq:59}) is depending on absolute state of the $i$th follower, and requiring the knowledge of the leaders' dynamics and the graph topology $\mathcal{G}$, which are not available for $i$th follower, therefore ${\hat{ \bar X}_i}(t) = {[{\xi_i}^T,{\hat{\eta_i}^T}]^T}$ should be used in place of ${\bar X_i}(t)$ to implement the optimal control (\ref{eq:59}) without requiring the knowledge of the leaders' dynamics, graph topology $\mathcal{G}$, and absolute state of the $i$th follower. By doing so, optimal control (\ref{eq:59}) becomes
\begin{align}
u_i^* = \bar K_i^*{\hat {\bar X}_i} \label{eq:68}
\end{align}
where ${\hat{ \bar X}_i}(t) = {[{\xi_i}^T,{\hat{\eta_i}^T}]^T}$ and $\bar K_i^*$ is obtained by (\ref{eq:60})-(\ref{eq:61}).

\begin{theorem} \label{thm3}
Consider the multi-agent system (\ref{eq:2})-(\ref{eq:5}) and the distributed adaptive observer (\ref{eq:35}) along with adaptation laws (\ref{eq:36}) and (\ref{eq:37}). Under Assumptions 1 - 6, Problem 2 and consequently Problem 1 are solved using optimal control policy (\ref{eq:68}) with $\bar K_i^*$ given by (\ref{eq:60}) and (\ref{eq:61}). As long as, ${\mu _i}$  and ${F_i}$ be designed as of Lemma 2, the discount factors ${\gamma _i}$ satisfy (\ref{eq:62}), and ${\beta _1}, {\beta _2}$, and  ${\beta _3}$ in (\ref{eq:35})-(\ref{eq:37}) be any positive constant.
\end{theorem}

\begin{proof}
Using (\ref{eq:2})-(\ref{eq:5}), (\ref{eq:14}), (\ref{eq:35}), and (\ref{eq:68}), we obtain
\begin{align}
\left[ {\begin{array}{*{20}{c}}
{\dot X}\\
{\dot \xi }\\
{\dot{ \hat \eta}  }
\end{array}} \right] = \left[ {\begin{array}{*{20}{c}}
A&{\bar BK_1^*}&{\bar BK_2^*}\\
{\mu F({L_1} \otimes {I_q})\bar C}&{A + \bar BK_1^* - \mu F({L_1} \otimes {I_q})\bar C}&{\bar BK_2^*}\\
{{0_{n\bar q \times N}}}&{{0_{n\bar q \times N}}}&{{{\bar{ \bar S}}_i} - ({L_1} \otimes \beta {I_{\bar q}})}
\end{array}} \right]\left[ {\begin{array}{*{20}{c}}
X\\
\xi \\
{\hat \eta }
\end{array}} \right] + \left[ {\begin{array}{*{20}{c}}
{{0_{N \times m\bar q}}}\\
{({L_2} \otimes {D_i})}\\
{ - ({L_2} \otimes \beta {I_{\bar q}})}
\end{array}} \right]\Omega  \label{eq:69}
\end{align}
where $\hat \eta = col({\hat \eta _1},...,{\hat \eta _n})$, $A = diag({A_1},...,{A_n})$, $\bar B = diag({B_1},...,{B_n})$, $\bar C = diag({C_1},...,{C_n})$, $N = \sum\limits_{i = 1}^n {{N_i}} $, $K_1^* = diag({K_1^1}^*,...,{K_n^1}^*)$, $K_2^* = diag({K_1^2}^*,...,{K_n^2}^*)$ and ${\bar {\bar S}_i} = {I_n} \otimes {S_i}$. For nonsingular matrix 
\begin{align}
T = \left[ {\begin{array}{*{20}{c}}
{{I_N}}&{{0_{N \times N}}}&{{0_{N \times n\bar q}}}\\
{{I_N}}&{{I_N}}&{{0_{N \times n\bar q}}}\\
{{0_{n\bar q \times N}}}&{{0_{n\bar q \times N}}}&{{I_{n\bar q}}}
\end{array}} \right],{\rm{ }} \nonumber
\end{align}
one can verify that (\ref{eq:69}) has a block-triangular structure and according to separation principle, the distributed observers and control gains in (\ref{eq:69}) can be designed separately. It is shown in Theorem 2 that, for any positive constant ${\beta _1}, {\beta _2}$, and ${\beta _3}$, ${S_i} \to S$, ${D_i} \to D$, ${\hat \eta _i} \to \omega _i^*$, $\forall i \in \mathcal{F}$, ${D_i}{\hat \eta _i} \to D\omega _i^*$ and thus for any full row rank matrix $D$, $( {{I_n} \otimes D} )\hat \eta  \to  - ( {L_1^{ - 1}{L_2} \otimes {I_q}} ){Y_R}$ asymptotically. Besides, it is shown in Lemma 2 that ${\xi _i} \to {x_i}$, ${C_i}{\xi _i} \to {y_i}$, $\forall i \in \mathcal{F}$ asymptotically. Finally, Lemma 5 shows that ${Y_F} \to ( {{I_n} \otimes D} )\hat \eta $ asymptotically, therefore ${Y_F} \to  - ( {L_1^{ - 1}{L_2} \otimes {I_q}} ){Y_R}$  asymptotically, which completes the proof.
\end{proof}

\begin{remark}
 Theorem 3 shows that one can design the observers gains appear in  (\ref{eq:14}), (\ref{eq:35})-(\ref{eq:37}) and the control gains $\bar K_i^*$ in (\ref{eq:68}) independently. As stated in Theorem 3, for $i \in \mathcal{F}$ , ${\mu _i}$ should be chosen sufficiently large, ${F_i}$ should be designed as (\ref{eq:18}) according to Lemma 2, and  ${\beta _1}$, ${\beta _2}$, and ${\beta _3}$ in (\ref{eq:35})-(\ref{eq:37}) are any positive constants. It can be seen from (\ref{eq:43}) and (\ref{eq:45}) that the larger ${\beta _1}$, ${\beta _2}$, and ${\beta _3}$ are , the faster the estimation errors of the distributed observer (\ref{eq:35})-(\ref{eq:37}) decays. It can also be seen from (\ref{eq:17}) that the larger ${\mu _i}$ for $i \in \mathcal{F}$ are, the faster the estimation errors decays for the observer (\ref{eq:14}). Choosing sufficiently large ${\mu _i}$ for $i \in F$, ${\beta _1}$, ${\beta _2}$, and ${\beta _3}$ makes the convergence of the observers sufficiently fast, and therefore, their effects on the control performance are negligible. The weight matrices ${Q_i}$  and ${W_i}$ for $i \in \mathcal{F}$ are design parameters and they should be chosen as symmetric positive definite matrices. The discount factors ${\gamma _i}$ for $i \in \mathcal{F}$ are used to guarantee that the performance functions is bounded for given control policies, and to do so, they should be chosen to satisfy condition (\ref{eq:62}) in Remark 12.
\end{remark}

\subsection{Distributed Optimal Output Containment Control: An Online and Optimal Solution}

In this subsection, an off-policy RL method\cite{Modares2015} is combined with actor-critic algorithm to solve discounted ARE (\ref{eq:61}) and learn the optimal gain (\ref{eq:60}) online in real time, using relative output measurements of followers with respect to their neighbors and the information which is broadcasted by neighbors, and without requiring complete knowledge of the leaders' dynamics. 

Since Value function (\ref{eq:57}) is in the form of a quadratic polynomial, the quadratic polynomial basis vector for the critic neural network (NN) of each follower is chosen as

\begin{align}
\sigma _i^c({\bar X_i}) = {[ {\begin{array}{*{20}{c}}
{\bar X_{i1}^2}&{{{\bar X}_{i1}}{{\bar X}_{i2}}}& \ldots &{\bar X_{iN_i^{{\sigma _c}}}^2}
\end{array}} ]^T} \in {\Re ^{{{N_i^{{\sigma _c}}(N_i^{{\sigma _c}} + 1)} \mathord{\left/
 {\vphantom {{N_i^{{\sigma _c}}(N_i^{{\sigma _c}} + 1)} 2}} \right.
 \kern-\nulldelimiterspace} 2}}}  \label{eq:70}
\end{align}
where ${\bar X_i} = {[{\bar X_{i1}},...,{\bar X_{iN_i^{{\sigma _c}}}}]^T}$ and $N_i^{{\sigma _c}} = {N_i} + \bar q$. Then, the optimal value function ${V_i}({\bar X_i})$ can be perfectly approximated by 
\begin{align}
{V_i}({\bar X_i}) = {\left( {\bar W_i^c} \right)^T}\sigma _i^c({\bar X_i}) \label{eq:71}
\end{align}
where $\bar W_i^c \in {\Re ^{M_i^{{\sigma _c}}}}$ is the optimal critic NN weight of $i$th follower and $M_i^{{\sigma _c}}= {{N_i^{{\sigma _c}}(N_i^{{\sigma _c}} + 1)} \mathord{\left/
 {\vphantom {{N_i^{{\sigma _c}}(N_i^{{\sigma _c}} + 1)} 2}} \right.
 \kern-\nulldelimiterspace} 2}$ is the number of hidden-layer neurons of the critic NN of $i$th follower.

By the same token, the optimal policy in (\ref{eq:59}) can be perfectly approximated by an actor NN for each follower in the form of
\begin{align}
{u_i}^*({\bar X_i}) = {\left( {\bar W_i^a} \right)^T}\sigma _i^a({\bar X_i}) \label{eq:72}
\end{align}
where ${\sigma _a}({\bar X_i}) = {[{\bar X_{i1}},...,{\bar X_{iM_i^{{\sigma _a}}}}]^T}$, $\bar W_i^a \in {\Re ^{M_i^{{\sigma _a}} \times {p_i}}}$ is the optimal actor NN weight of $i$th follower  and $M_i^{{\sigma _a}}=N_i^{{\sigma _c}}$  is the number of hidden-layer neurons of the actor NN of $i$th follower, for $i \in \mathcal{F}$. The ideal weights of the critic NNs, i.e., ${\bar W_i^c}$, and the ideal weights of the actor NNs, i.e., ${\bar W_i^a}$, are unknown and must be estimated.

Let the value function corresponding to $\hat u_i^k$ be written as follows

\begin{align}
\begin{array}{l}
\hat V_i^k({{\bar X}_i}) = \int\limits_t^\infty  {{e^{ - {\gamma _i}(\tau  - t)}}{{\bar X}_i}^T({{\bar {\bar C}}_i}^T{Q_i}{{\bar {\bar C}}_i} + \bar K{{_i^k}^T}{W_i}\bar K_i^k){{\bar X}_i}d\tau } = {{\bar X}_i}^T\Psi _i^k{{\bar X}_i}
\end{array} \label{eq:73}
\end{align}
where 
\begin{align}
\hat u_i^k = \bar K_i^k{\bar X_i}  \label{eq:74}
\end{align}
with
\begin{align}
{\bar K_i}^k = - W_i^{- 1}\bar B_i^T{\Psi _i}^k \label{eq:75}
\end{align}
is the estimation of ${u_i}^*$ in the $k$th iteration, and ${\Psi _i}^k$ satisfies the discounted ARE (\ref{eq:80}).

Let the estimation of $\bar W_i^c$ in the $k$th iteration be denoted by $\hat {\bar W}{{_i^c}^k}$. Then, the value function (\ref{eq:71}) and its gradient $\nabla V_i({\bar X_i}) = {{\partial V_i} \mathord{\left/
 {\vphantom {{\partial \hat V_i^k} {\partial {{\bar X}_i}}}} \right.
 \kern-\nulldelimiterspace} {\partial {{\bar X}_i}}}$ are approximated as

\begin{align}
\hat V_i^k({\bar X_i}) = {({\hat {\bar W}{{_i^c}^k}})^T}\sigma _i^c({\bar X_i})  \label{eq:76}
\end{align}
\begin{align}
\nabla \hat V_i^k({\bar X_i}) = {(\nabla \sigma _i^c({\bar X_i}))^T}{\hat {\bar W}{{_i^c}^k}}  \label{eq:77}
\end{align}

Correspondingly, let $\hat {\bar W}{{_i^a}^k}$ denote the estimation of $\bar W_i^a$ in the $k$th iteration. Then, the estimation of ${\rm{u}}_i^*({\bar X_i})$ in the $k$th iteration is 

\begin{align}
\hat u_i^k({\bar X_i}) = {(\hat {\bar W}{{_i^a}^k})^T}\sigma _i^a({\bar X_i})  \label{eq:78}
\end{align}

One can now write (\ref{eq:53}) in following augmented form as
\begin{align}
{\dot {{\bar X}}_i} = P_i^k{\bar X_i} + {\bar B_i}({u_i} - \hat u_i^k)  \label{eq:79}
\end{align} 
where $P_i^k = {P_i} + {\bar B_i}\bar K_i^k$ and ${u_i}$ is a  behavior policy, which is an admissible policy, applied to the $i$th follower and $\hat u_i^k$ is the estimation of target policy, which is the optimal policy, in the $k$th iteration. 

Using (\ref{eq:73})-(\ref{eq:75}), (\ref{eq:79}) and some manipulation, the discounted ARE (\ref{eq:61}) becomes
\begin{align}
\Psi _i^kP_i^k + {(P_i^k)^T}\Psi _i^k - {\gamma _i}\Psi _i^k + {\bar {\bar C}_i}^T{Q_i}{\bar {\bar C}_i} - {(\bar K_i^k)^T}{W_i}\bar K_i^k = 0 \label{eq:80}
\end{align}

Consider the control policy update law as follows
\begin{align}
\begin{array}{l}
u_i^{k + 1} =  - W_i^{ - 1}\bar B_i^T\Psi _i^k{{\bar X}_i}
\end{array} \label{eq:81}
\end{align}

The derivative of (\ref{eq:73}) with respect to time along the system dynamics (\ref{eq:79}) can be derived as
\begin{align}
\begin{array}{l}
\dot {\hat V}_i^k({{\bar X}_i}) = {{\bar X}_i}^T(\Psi _i^kP_i^k + {(P_i^k)^T}\Psi _i^k){{\bar X}_i}\, + 2{{\bar X}_i}^T\Psi _i^k{{\bar B}_i}({u_i} - \hat u_i^k)\\
\,\,\,\,\,\,\,\,\,\,\,\,\,\,\,\,\,\,\, = {\gamma _i}\hat V_i^k({{\bar X}_i}) - {{\bar X}_i}^T{{\bar {\bar C}}_i}^T{Q_i}{{\bar {\bar C}}_i}{{\bar X}_i} - {(\hat u_i^k)^T}{W_i}\hat u_i^k - 2{(\hat u_i^{k + 1})^T}{W_i}({u_i} - \hat u_i^k)
\end{array} \label{eq:82}
\end{align}

Multiplying ${e^{ - {\gamma _i}(\tau  - t)}}$ to the both sides of (\ref{eq:82}), and integrating both sides on the time interval of $\left[ {t;{\rm{ }}t + T} \right]$, one can obtain the off-policy Bellman equation as follows

\begin{align}
\begin{array}{l}
{e^{ - {\gamma _i}T}}\hat V_i^k({{\bar X}_i}(t + T)) - \hat V_i^k({{\bar X}_i}(t)) = \int\limits_t^{t + T} {\frac{d}{{d\tau }}(\,{e^{ - {\gamma _i}(\tau  - t)}}\hat V_i^k({{\bar X}_i}))d\tau }  = \\
\qquad \qquad \qquad \qquad \int\limits_t^{t + T} {{e^{ - {\gamma _i}(\tau  - t)}}( - {{\bar X}_i}^T{{\bar {\bar C}}_i}^T{Q_i}{{\bar {\bar C}}_i}{{\bar X}_i}\, - {{(\hat u_i^k)}^T}{W_i}\hat u_i^k)d\tau }  + \,\int\limits_t^{t + T} {{e^{ - {\gamma _i}(\tau  - t)}}( - 2{{(\hat u_i^{k + 1})}^T}{W_i}({u_i} - \hat u_i^k))d\tau } 
\end{array}  \label{eq:83}
\end{align}

\begin{remark}
Note that to calculate the first integrand in the right-hand-side of (\ref{eq:83}), the knowledge of the leaders' dynamic $D$ is required due to dependency of ${\bar {\bar C}_i}^T{Q_i}{\bar {\bar C}_i}$ to it, i.e., ${\bar X_i}^T{\bar {\bar C}_i}^T{Q_i}{\bar {\bar C}_i}{\bar X_i} = {\left( {{y_i} - D\omega _i^*} \right)^T}{Q_i}({y_i} - D\omega _i^*)$. To obviate this requirement, the term of ${\bar X_i}^T{\bar {\bar C}_i}^T{Q_i}{\bar{ \bar C}_i}{\bar X_i}$, in first integrand in the right-hand-side of (\ref{eq:83}), is replaced by ${\left( {{{\hat y}_i} - {y_{0i}}} \right)^T}Q_i({\hat y_i} - {y_{0i}})$
, where ${y_{0i}}$ is the estimation of convex combination of the leaders' outputs, defined in (\ref{eq:38}) and ${\hat y_i} = {C_i}{\xi _i}$ is the output estimation of $i$th follower.
\end{remark}

Define $\upsilon _i^k = {u_i} - \hat u_i^k$, exploiting the critic NN (\ref{eq:76}), actor NN (\ref{eq:78}), weighted matrices of ${W_i} = diag(w_i^1,...,w_i^{{p_i}})$, and Remark 14, (\ref{eq:83}) can be put into the following form
\begin{align}
\begin{array}{l}
e_i^k = {(\hat {\bar W}{{_i^c}^k})^T}[{e^{ - {\gamma _i}T}}\sigma _i^c({{\bar X}_i}(t + T)) - \sigma _i^c({{\bar X}_i}(t))] + 2\sum\limits_{j = 1}^{{p_i}} {{{( {\hat {\bar W}{{_{i,j}^a}^{k + 1}}} )}^T}w_i^j\int\limits_t^{t + T} {{e^{ - {\gamma _i}(\tau  - t)}}\sigma _i^a({{\bar X}_i})\upsilon _{i,j}^kd\tau } } \\
\qquad  - \int\limits_t^{t + T} {{e^{ - {\gamma _i}(\tau  - t)}}( - {{( {{C_i}{\xi _i} - {D_i}{{\hat \eta }_i}} )}^T}{Q_i}({C_i}{\xi _i} - {D_i}{{\hat \eta }_i})\, - \sigma _i^a{{({{\bar X}_i})}^T}\hat {\bar W}{{_i^a}^k}{W_i}{{(\hat {\bar W}{{_i^a}^k})}^T}\sigma _i^a({{\bar X}_i}))d\tau } 
\end{array}\label{eq:84}
\end{align}
where $\hat {\bar W}{{_i^a}^{k + 1}} = {[ {\begin{array}{*{20}{c}}
{\hat {\bar W}{{_{i,1}^a}^{k + 1}}}& \cdots &{\hat {\bar W}{{_{i,{p_i}}^a}^{k + 1}}}
\end{array}}]}$, $\upsilon _i^k = {[ {\begin{array}{*{20}{c}}
{\upsilon _{i,1}^k}& \cdots &{\upsilon _{i,{p_i}}^k}
\end{array}} ]^T}$, and $e_i^k$ is the Bellman approximation error which should be minimized in order to drive the critic NN weights and actor NN weights toward their ideal values, i.e., $\hat {\bar W}{{_i^c}^k} \to \bar W_i^c$ and $\hat {\bar W}{{_i^a}^{k + 1}} \to \bar W_i^a$.

Rearranging (\ref{eq:84}), the Bellman equation (\ref{eq:83}) can be reformulated as follows 
\begin{align}
{Y_i}^k(t) = {({\hat W_i}^k)^T}\Xi _i^k(t)-e_i^k \label{eq:85}
\end{align}
where 
\begin{align}
{Y_i}^k(t) = \int\limits_t^{t + T} {{e^{ - {\gamma _i}(\tau  - t)}}( - {( {{C_i}{\xi _i} - {D_i}{{\hat \eta }_i}} )^T}{Q_i}({C_i}{\xi _i} - {D_i}{{\hat \eta }_i})\, - \sigma _i^a{{({{\bar X}_i})}^T}\hat {\bar W}{{_i^a}^k}{W_i}{{(\hat {\bar W}{{_i^a}^k})}^T}\sigma _i^a({{\bar X}_i}))d\tau } \label{eq:86}
\end{align}

\begin{align}
{\hat W_i}^k = [ {{(\hat {\bar W}{{_i^c}^k})^T},{( {\hat {\bar W}{{_{i,1}^a}^{k + 1}}})^T},...,{( {\hat {\bar W}{{_{i,{p_i}}^a}^{k + 1}}})^T}} ]^T \in {\Re ^{M_i^{{\sigma _c}} + {p_i} \times M_i^{{\sigma _a}}}} \label{eq:87}
\end{align}
\begin{align}
\Xi _i^k(t) = \left[ {\begin{array}{*{20}{c}}
{{e^{ - {\gamma _i}T}}\sigma _i^c({{\bar X}_i}(t + T)) - \sigma _i^c({{\bar X}_i}(t))}\\
{2w_i^1\int\limits_t^{t + T} {{e^{ - {\gamma _i}(\tau  - t)}}\sigma _i^a({{\bar X}_i})\upsilon _{i,1}^kd\tau } }\\
 \vdots \\
{2w_i^{{p_i}}\int\limits_t^{t + T} {{e^{ - {\gamma _i}(\tau  - t)}}\sigma _i^a({{\bar X}_i})\upsilon _{i,{p_i}}^kd\tau } }
\end{array}} \right]\in {\Re ^{M_i^{{\sigma _c}} + {p_i} \times M_i^{{\sigma _a}}}} \label{eq:88}
\end{align}

Assume that ${Y_i}^k(t)$ and $\Xi _i^k(t)$ are collected at $\mathbb{n}_i  \ge M_i^{{\sigma _c}} + {p_i} \times M_i^{{\sigma _a}}$ points ${t_1}$ to ${t_{\mathbb{n}_i}}$
, over the same time interval $\left[ {t;t + T} \right]$. 
Using least square method in average sense, one has
\begin{align}
{\hat W_i}^k = {( {\bar \Xi _i^k{( {\bar \Xi _i^k} )^T}} )^{ - 1}}\bar \Xi _i^k{\bar Y_i}^k \label{eq:89}
\end{align}
where 
\begin{align}
\bar \Xi _i^k = [\Xi _i^k({t_1}),...,\Xi _i^k({t_{\mathbb{n}_i}})] \label{eq:90}
\end{align}
\begin{align}
{\bar Y_i}^k = [{Y_i}^k({t_1}),...,{Y_i}^k({t_{\mathbb{n}_i} })]^T \label{eq:91}
\end{align}

\begin{remark}
Note that to calculate the information in (\ref{eq:86}) and (\ref{eq:88}) for each follower, the knowledge of graph topology $\mathcal{G}$ and absolute states of followers are required due to the dependency of ${\bar X_i}(t) = [{x_i}^T,\omega {{_i^*}^T}]^T$ in (\ref{eq:86}) and (\ref{eq:88}) to them. To obviate these requirements, in Algorithm 2, the term of ${\hat{ \bar X}_i}(t) = {[{\xi_i}^T,{\hat{\eta_i}^T}]^T}$ will be used in place of ${\bar X_i}(t) = {[{x_i}^T,{\omega_i^*}^T]^T}$ in (\ref{eq:86}) and (\ref{eq:88}). 
\end{remark}

To this end, by using ${\hat{ \bar X}_i}(t) = {[{\xi_i}^T,{\hat{\eta_i}^T}]^T}$ in place of ${\bar X_i}(t) = {[{x_i}^T,{\omega_i^*}^T]^T}$, (\ref{eq:89}) can be rewritten as follows

\begin{align}
{\hat W_i}^k = {( {\hat {\bar \Xi} _i^k{( {\hat {\bar \Xi} _i^k})^T}} )^{ - 1}} \hat {\bar \Xi} _i^k \hat {\bar Y}_i^k \label{eq:92}
\end{align}
where $\hat W_i^k$ is given in (\ref{eq:87}), and
\begin{align}
\hat {\bar \Xi} _i^k = [\hat \Xi _i^k({t_1}),...,\hat \Xi _i^k({t_{\mathbb{n}_i}})] \label{eq:93}
\end{align}
\begin{align}
{\hat {\bar Y}_i}^k = [\hat {Y_i}^k({t_1}),...,\hat {Y_i}^k({t_{\mathbb{n}_i} })]^T \label{eq:94}
\end{align}
\begin{align}
{\hat Y_i}^k(t) = \int\limits_t^{t + T} {{e^{ - {\gamma _i}(\tau  - t)}}( - {( {{C_i}{\xi _i} - {D_i}{{\hat \eta }_i}} )^T}{Q_i}({C_i}{\xi _i} - {D_i}{{\hat \eta }_i})\, - \sigma _i^a{({{\hat{ \bar X}_i}})^T}\hat {\bar W}{{_i^a}^k}{W_i}{(\hat {\bar W}{{_i^a}^k})^T}\sigma _i^a({{\hat{ \bar X}_i}}))d\tau } \label{eq:95}
\end{align}
\begin{align}
\hat \Xi _i^k(t) = \left[ {\begin{array}{*{20}{c}}
{{e^{ - {\gamma _i}T}}\sigma _i^c({\hat{ \bar X}_i}(t + T)) - \sigma _i^c({\hat{ \bar X}_i}(t))}\\
{2w_i^1\int\limits_t^{t + T} {{e^{ - {\gamma _i}(\tau  - t)}}\sigma _i^a({\hat{ \bar X}_i})\upsilon _{i,1}^kd\tau } }\\
 \vdots \\
{2w_i^{{p_i}}\int\limits_t^{t + T} {{e^{ - {\gamma _i}(\tau  - t)}}\sigma _i^a({\hat{ \bar X}_i})\upsilon _{i,{p_i}}^kd\tau } }
\end{array}} \right] \label{eq:96}
\end{align}

To simultaneously solve the ARE (\ref{eq:83}) and find the optimal policy (\ref{eq:72}) online, i.e., solve Problem 2, the off-policy RL algorithm is given in Algorithm 1. Note that, Algorithm 1 solve Problem 2 without requiring any knowledge about graph topology $\mathcal{G}$, and absolute states or outputs of followers. Moreover, Algorithm 1 solve problem 2 without the restrictive assumption of which all followers should have knowledge about the leaders' dynamics $S$ and $D$.
\\
\\
\\
\\
\\
\\
\noindent\rule{17.8cm}{0.4pt} \\
\textbf{Algorithm 1: Online Off-policy RL Algorithm} \\
\rule{17.8cm}{0.2pt} \\
1) Start with an admissible control policy $u_i^\kappa  = K_i^\kappa {\hat {\bar X}_i} + \varepsilon $, where $\varepsilon $ is the exploration noise and $\kappa  = 0$, and collect required information, i.e., $\hat {\bar \Xi}_i^k$ and $\hat {\bar Y}_i^k$. \\
2) Solve the least square problem (\ref{eq:92}) to obtain $\hat {\bar W}{{_i^c}^k}$ and $\hat {\bar W}{{_i^a}^{k + 1}}$ simultaneously. \\
3) Let $\kappa  = \kappa  + 1$, and repeat step 2 until $\| {{\hat {W}_i}^k - {{\hat {W}}_i}^{k - 1}} \| \le \tau $, where $\tau $ is a small predefined positive constant.\\
4) On convergence set $u_i^* = \hat u_i^k = {({\hat {\bar W}{{_i^a}^k}})^T}\sigma _i^a({\hat {\bar X}_i})$ as the optimal control policy. \\
\rule{17.8cm}{0.2pt} \\

\begin{remark}
Note that, due to the terms ${\hat \eta _i}$ and  ${\xi _i}$ in (\ref{eq:94}) and (\ref{eq:96}), information about $\hat {\bar \Xi}_i^k$ and $\hat {\bar Y}_i^k$, which are required in Algorithm 1, cannot be collected and used from the beginning of the learning process. However, once the distributed observers in (\ref{eq:14}) and (\ref{eq:35})-(\ref{eq:37}) converge, information about $\hat {\bar \Xi}_i^k$ and $\hat {\bar Y}_i^k$ is collected, and then, it is used by Algorithm 1. Note also that, as mentioned in Remark 13, the observers gains appear in  (\ref{eq:14}), (\ref{eq:35})-(\ref{eq:37}), and the control gains ${\bar K_i^*}$ in (\ref{eq:68}) can be designed  independently, and choosing sufficiently large ${\mu _i}$ for $i \in \mathcal{F}$, ${\beta _1}$, ${\beta _2}$, and ${\beta _3}$ makes the convergence of the observers sufficiently fast, and therefore, the effect of this issue on Algorithm 1 is negligible.
\end{remark}

\section{Simulation Results}\label{sec6}

In this section, an output containment control example is given to validate the proposed approach. Consider the multi-agent system consists of four heterogeneous followers and three leaders. The leaders' dynamics are given as
\begin{align}
S = \left[ {\begin{array}{*{20}{c}}
1&{ - 3}\\
1&{ - 1}
\end{array}} \right], 
D = \left[ {\begin{array}{*{20}{c}}
1&0\\
0&1
\end{array}} \right]  \label{eq:97}
\end{align}

The initial leaders' state vectors are chosen as ${\omega _5}(0) = {[{\begin{array}{*{20}{c}}
2&1
\end{array}}]^T}$, ${\omega _6}(0) = {[ {\begin{array}{*{20}{c}}
{ - 1}&1
\end{array}}]^T}$, and ${\omega _7}(0) = {[ {\begin{array}{*{20}{c}}
{0.4}&{0.4}
\end{array}}]^T}$.

The dynamics of heterogeneous followers are given as
\begin{align}
\begin{array}{l}
{A_1} = \left[ {\begin{array}{*{20}{c}}
{ - 1}&0&0\\
0&3&0\\
0&3&2
\end{array}} \right],{B_1} = \left[ {\begin{array}{*{20}{c}}
4\\
1\\
1
\end{array}} \right],C_1^T = 
\left[ \begin{array}{cc}
0 & 0  \\
1 & 0  \\
0 & 1  \end{array} \right] \\
\\
{A_2} = \left[ {\begin{array}{*{20}{c}}
1&{ - 1}\\
1&0
\end{array}} \right],{B_2} = \left[ {\begin{array}{*{20}{c}}
{ - 2}\\
{ - 1}
\end{array}} \right],C_2^T = \left[ {\begin{array}{*{20}{c}}
1&0\\
0&1
\end{array}} \right] \\
\\
{A_3} = \left[ {\begin{array}{*{20}{c}}
2&0\\
2&2
\end{array}} \right],{B_3} = \left[ {\begin{array}{*{20}{c}}
{ - 1}\\
{ - 1}
\end{array}} \right],C_3^T = \left[ {\begin{array}{*{20}{c}}
1&0\\
0&1
\end{array}} \right]\\
\\
{A_4} = \left[ {\begin{array}{*{20}{c}}
{ - 1}&0&0\\
0&2&{ - 1}\\
0&3&3
\end{array}} \right],{B_4} = \left[ {\begin{array}{*{20}{c}}
5\\
1\\
2
\end{array}} \right],C_4^T = 
\left[ \begin{array}{cc}
0 & 0  \\
1 & 0  \\
0 & 1  \end{array} \right]. \\
\end{array} \label{eq:98}
\end{align}

The communication graph among the agents is given in Fig.~\ref{fig:1}, where nodes 5, 6, and 7 represent the leaders and other nodes represent four heterogeneous agents. Moreover, all the communication weights, in Fig.~\ref{fig:1}, are chosen to be one. It can be verified that Assumptions 1-5 are satisfied.

 \begin{figure}[!t]
  \centering{\includegraphics[width=3 in]{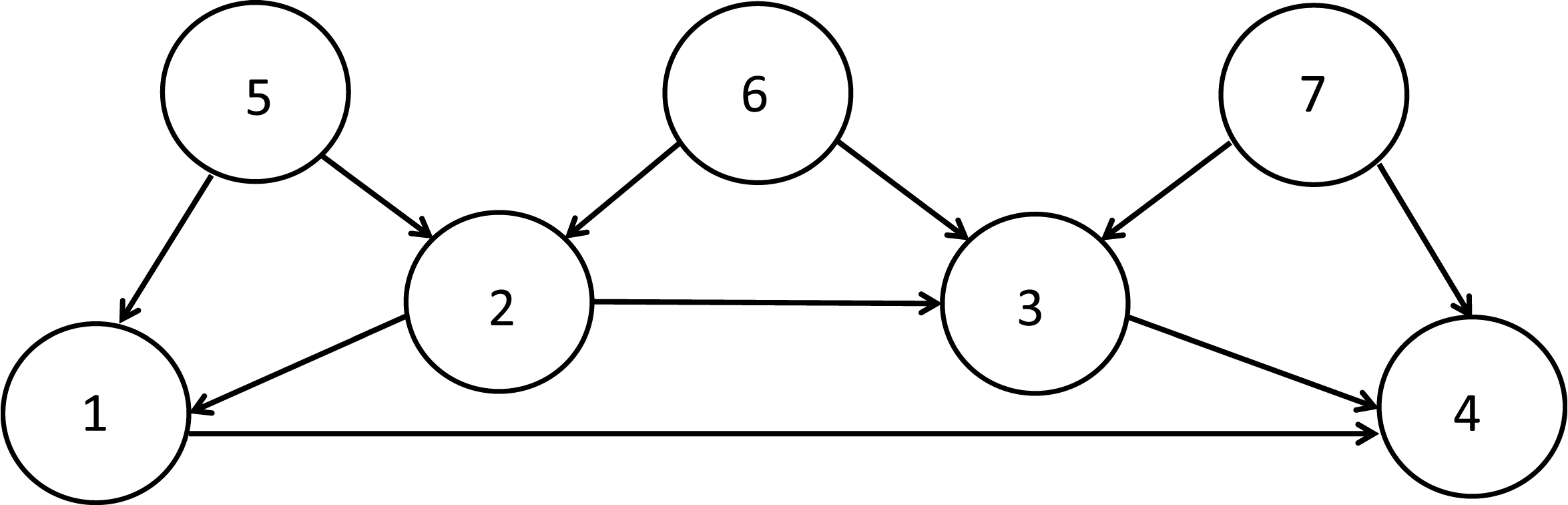}}
  \caption{The multi-agent systems communication graph.}\label{fig:1}
 \end{figure}
  
 The distributed observer (\ref{eq:14}) is implemented for  $i=1,2,3,4$. The observer matrices and gains are set as 
 \begin{align}
 {\Phi _1} = {I_3},{E_1} = \left[ {\begin{array}{*{20}{c}}
 2&0&0\\
 0&1&0\\
 0&0&1
 \end{array}} \right],{R_1} = \left[ {\begin{array}{*{20}{c}}
 {{\rm{0}}{\rm{.19}}}&{{\rm{ - 0}}{\rm{.11}}}\\
 {{\rm{ - 0}}{\rm{.11}}}&{{\rm{0}}{\rm{.26}}}
 \end{array}} \right] \nonumber 
  \end{align} 
  \begin{align}
 &{\Phi _2} = {I_2},{E_2} = {I_2},{R_2} = \left[ {\begin{array}{*{20}{c}}
   {0.33}&0\\
   0&1
   \end{array}} \right] \nonumber \\
  \nonumber \\
 &{\Phi _3} = {I_2},{E_3} = {I_2},{R_3} = \left[ {\begin{array}{*{20}{c}}
 {{\rm{0}}{\rm{.23}}}&{{\rm{ - 0}}{\rm{.09}}}\\
 {{\rm{ - 0}}{\rm{.09}}}&{{\rm{0}}{\rm{.23}}}
 \end{array}} \right] \label{eq:99}  \\ 
 \nonumber\\
 &{\Phi _4} = {I_3},{E_4} = \left[ {\begin{array}{*{20}{c}}
 2&0&0\\
 0&1&0\\
 0&0&1
 \end{array}} \right],{R_4} = \left[ {\begin{array}{*{20}{c}}
 {{\rm{0}}{\rm{.22}}}&{{\rm{ - 0}}{\rm{.06}}}\\
 {{\rm{  - 0}}{\rm{.06}}}&{{\rm{0}}{\rm{.16}}}
 \end{array}} \right], \nonumber
 \end{align} 
 and ${\mu _i} = 1$, for  $i=1,2,3,4$. The error between state observer (14) and followers' states is given in Fig.~\ref{fig:2}. It can be observed from Fig.~\ref{fig:2} that the error between state observers and the followers' states converges to zero.

 \begin{figure}[!t]
 \centering{\includegraphics[width=3.6in]{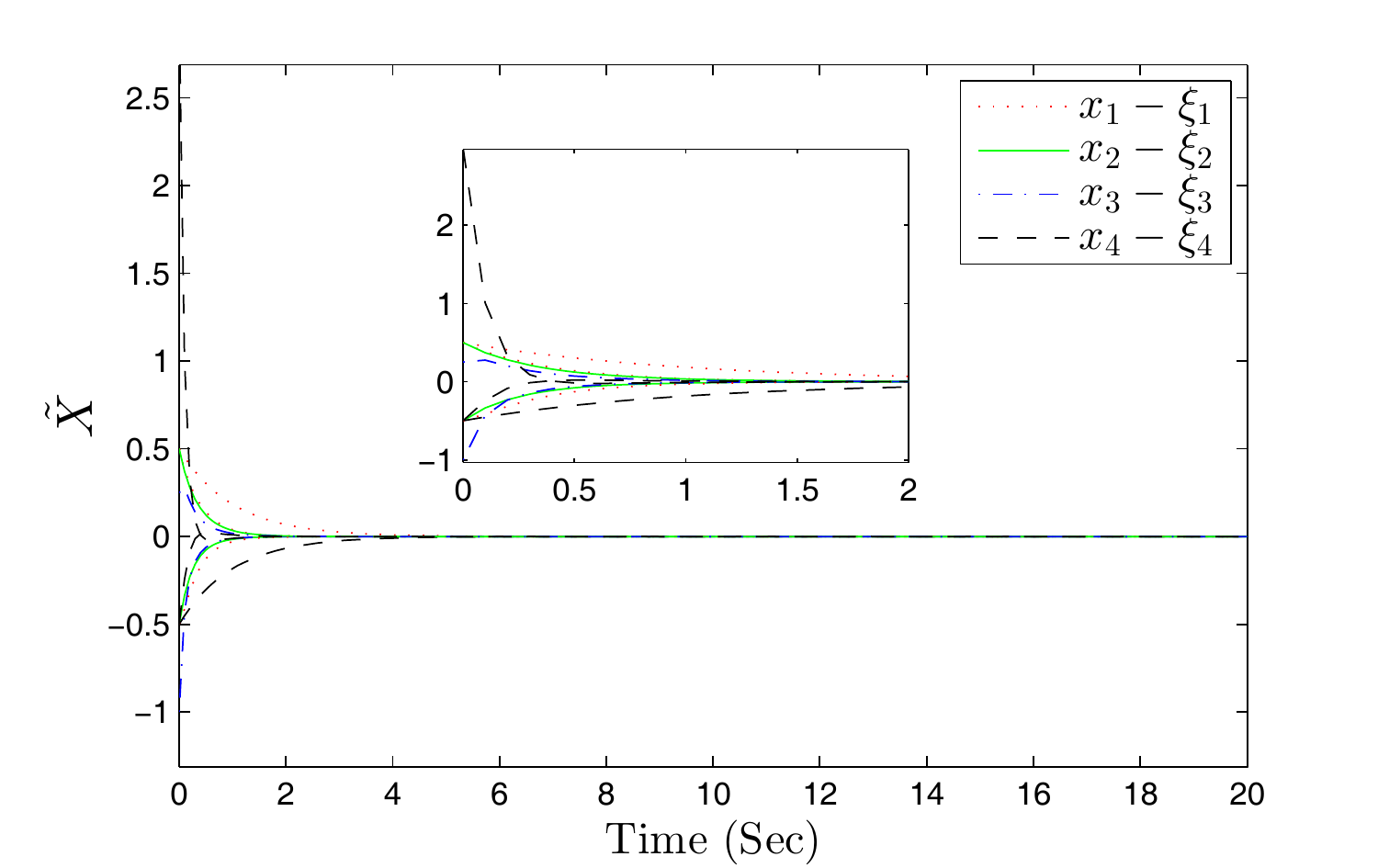}}
 \caption{Error between state observer (\ref{eq:14}) and followers' states.}  \label{fig:2}
 \end{figure}

 The distributed adaptive observer (\ref{eq:35}) along with adaptation laws (\ref{eq:36}) and (\ref{eq:37}) are also implemented for $i=1,2,3,4$. The observer and adaptive laws gains are set as ${\beta _1} = 3$, ${\beta _2} = 10$, and ${\beta _3} = 3$. It can be seen from Fig.~\ref{fig:3} that the error between observers (\ref{eq:35}) and the convex combination of the leaders' states converges to zero, i.e., the adaptive observers (\ref{eq:35}) converge to the convex combination of the leaders' states.
  \begin{figure}[!t]
  \centering{\includegraphics[width=3.6in]{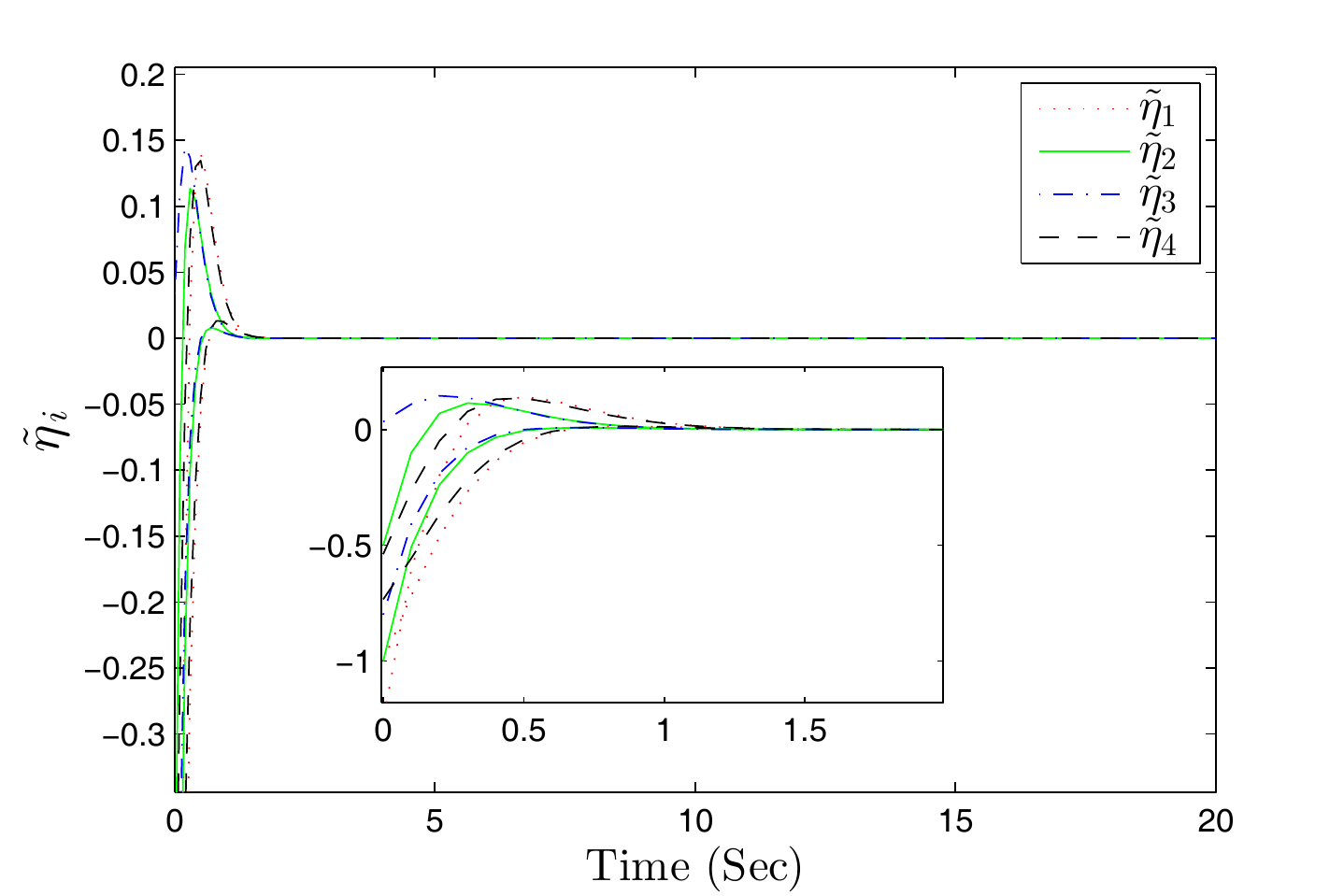}}
  \caption{Error between the adaptive observer (\ref{eq:35}) and the convex combination of the leaders' states.} \label{fig:3}
  \end{figure}
    
 By solving the output regulator equation (\ref{eq:6}), one can obtain
  \begin{align}
 \begin{array}{l}
 {\Pi _1} = \left[ {\begin{array}{*{20}{c}}
 4&{ - 16}\\
 {\begin{array}{*{20}{c}}
 1\\
 0
 \end{array}}&{\begin{array}{*{20}{c}}
 0\\
 1
 \end{array}}
 \end{array}} \right],{\Gamma _1} = \left[ {\begin{array}{*{20}{c}}
 { - 2}&{ - 3}
 \end{array}} \right]\\
 \\
 {\Pi _2} = {I_2},{\Gamma _2} = \left[ {\begin{array}{*{20}{c}}
 0&1
 \end{array}} \right]\\
 \\
 {\Pi _3} = {I_2},{\Gamma _3} = \left[ {\begin{array}{*{20}{c}}
 1&3
 \end{array}} \right]\\
 \\
 {\Pi _4} = \left[ {\begin{array}{*{20}{c}}
 {{\rm{3}}{\rm{.33}}}&{ - {\rm{11}}{\rm{.66}}}\\
 {\begin{array}{*{20}{c}}
 1\\
 0
 \end{array}}&{\begin{array}{*{20}{c}}
 0\\
 1
 \end{array}}
 \end{array}} \right],{\Gamma _4} = \left[ {\begin{array}{*{20}{c}}
 { - 1}&{ - 2}
 \end{array}} \right]. \label{eq:100}
 \end{array}
 \end{align} 

The weight matrices $Q_i$, $W_i$, and the discount factors $\gamma_i$ for $i=1,2,3,4$ are chosen as
\begin{align}
&{Q_1} = 10{I_3},{W_1} = 10,{\gamma _1} = 0.01 \nonumber \\
&{Q_2} = 10{I_2},{W_2} = 10,{\gamma _2} = 0.01 \nonumber \\
&{Q_3} = 10{I_2},{W_3} = 10,{\gamma _3} = 0.01 \nonumber \\
&{Q_4} = 10{I_3},{W_4} = 10,{\gamma _4} = 0.01.  \label{eq:101}
\end{align} 

Using (\ref{eq:98}), (\ref{eq:101}), and LQR method, the optimal dynamic output feedback gains are 
 
 \begin{align}
 &K_1^1 = \left[ {\begin{array}{*{20}{c}}
 {{{ - 0}}{{.13}}}&{{{7}}{{.91}}}&{{{ - 20}}{{.65}}}
 \end{array}} \right] \nonumber \\
 &K_2^1 = \left[ {\begin{array}{*{20}{c}}
 {\rm{1}}&{{\rm{1}}{\rm{.78}}}
 \end{array}} \right] \nonumber \\
 &K_3^1 = \left[ {\begin{array}{*{20}{c}}
 {{\rm{ - 0}}{\rm{.23}}}&{{\rm{8}}{\rm{.61}}}
 \end{array}} \right] \nonumber \\
 &K_4^1 = \left[ {\begin{array}{*{20}{c}}
 {{{ - 0}}{{.29}}}&{{{7}}{{.02}}}&{{{ - 10}}{{.1}}}
 \end{array}} \right].  \label{eq:102}
 \end{align}
  
 We now use Algorithm 1 to find the optimal dynamic output feedback control online in real-time. To do so, time interval is set as $T = 0.5$ Sec, and $N_1^{\sigma _c} = 5$,  $N_2^{{\sigma _c}} = 4$, $N_3^{{\sigma _c}} = 4$, $N_4^{{\sigma _c}} = 5$, $M_1^{{\sigma _c}} = 15$, $M_2^{{\sigma _c}} = 10$, $M_3^{{\sigma _c}} = 10$, $M_4^{{\sigma _c}} = 15$, $M_1^{{\sigma _a}} = 5$, $M_2^{{\sigma _a}} = 4$, $M_3^{{\sigma _a}} = 4$,  $M_4^{{\sigma _a}} = 5$, $\sigma _i^c({\hat {\bar X}_i}) = {[ {\begin{array}{*{20}{c}}
 {\hat {\bar X}_{i1}^2}&{{{\hat {\bar X}}_{i1}}{{\hat {\bar X}}_{i2}}}& \ldots &{\hat {\bar X}_{iN_i^{{\sigma _c}}}^2}
 \end{array}}]^T}$, ${\sigma_i ^a}({\hat {\bar X}_i}) = {[ {\begin{array}{*{20}{c}}
 {{{\hat {\bar X}}_{i1}}}& \cdots &{{{\hat {\bar X}}_{iM_i^{{\sigma _a}}}}}
 \end{array}} ]^T}$, for $i=1,2,3,4$.  
 
   \begin{figure}[!t]
    \centering{\includegraphics[width=3.6in]{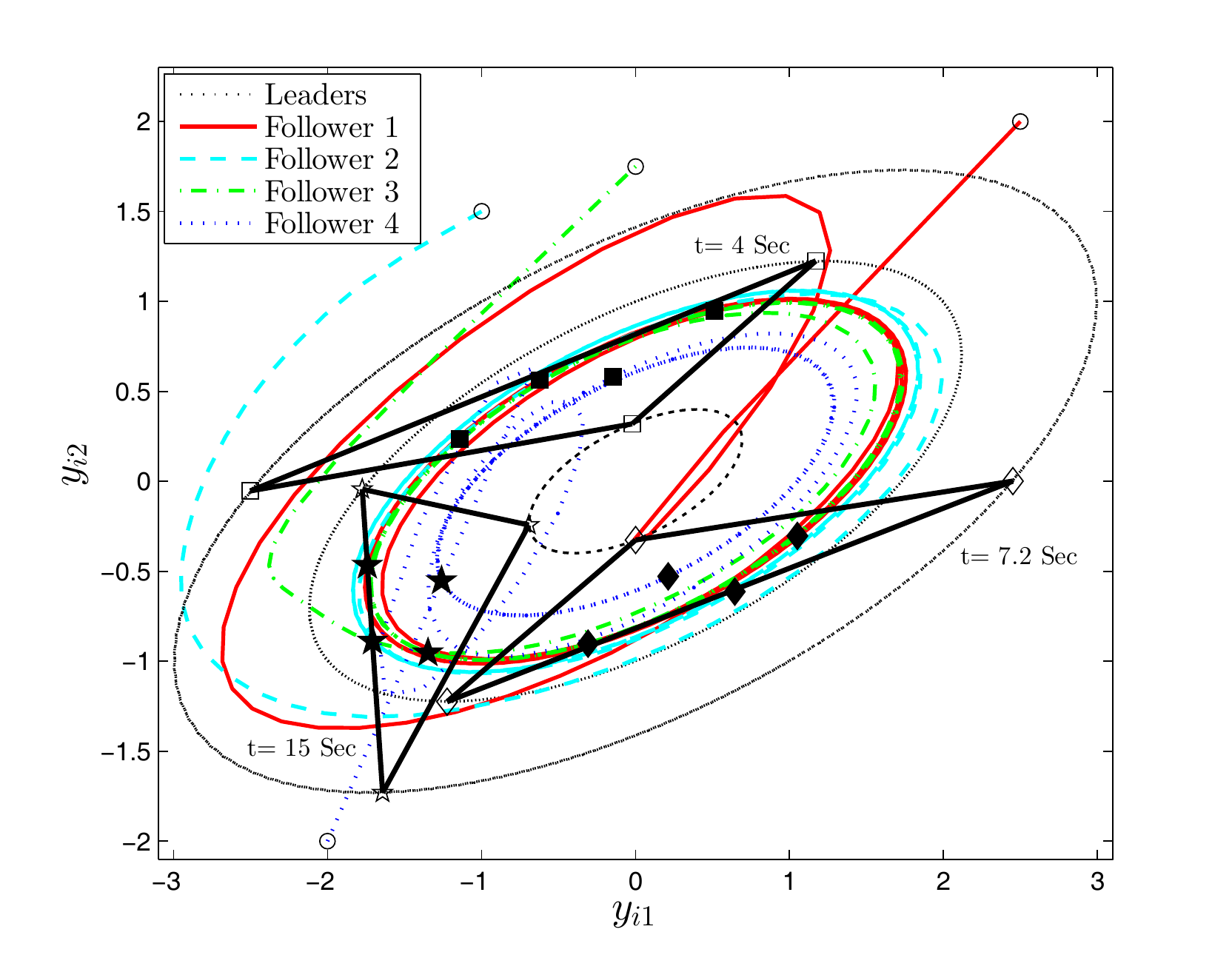}}
    \caption{The outputs of the agents using Algorithm 1, along with adaptive observer (\ref{eq:35}) and state observer (\ref{eq:14}).}    \label{fig:4}
    \end{figure}
   
    \begin{figure}[!t]
    \centering{\includegraphics[width=3.6in]{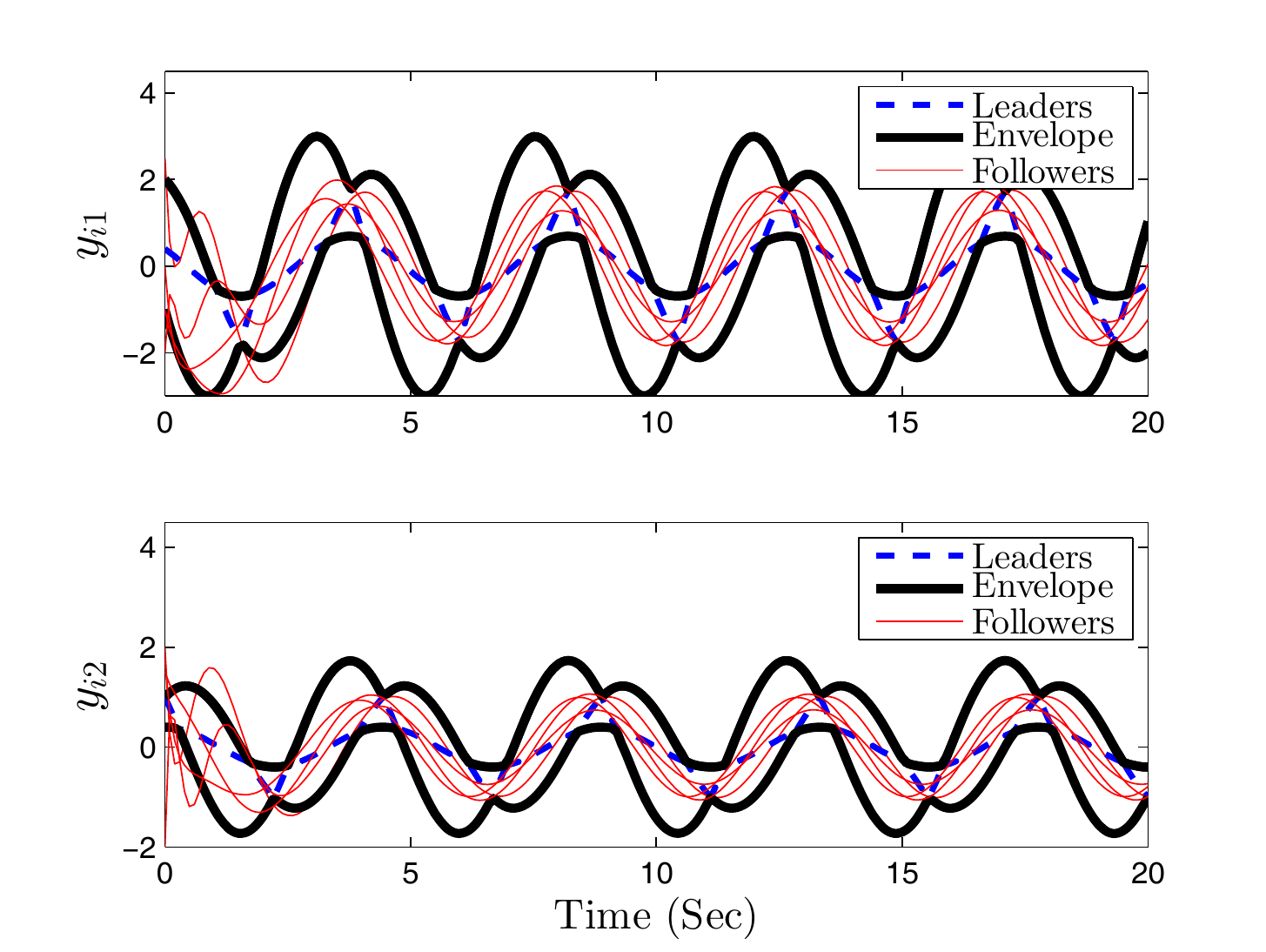}}
    \caption{ The outputs of all four heterogeneous agents.} \label{fig:5}
    \end{figure}

  \begin{figure}[!t]
   \centering{\includegraphics[width=3.6in]{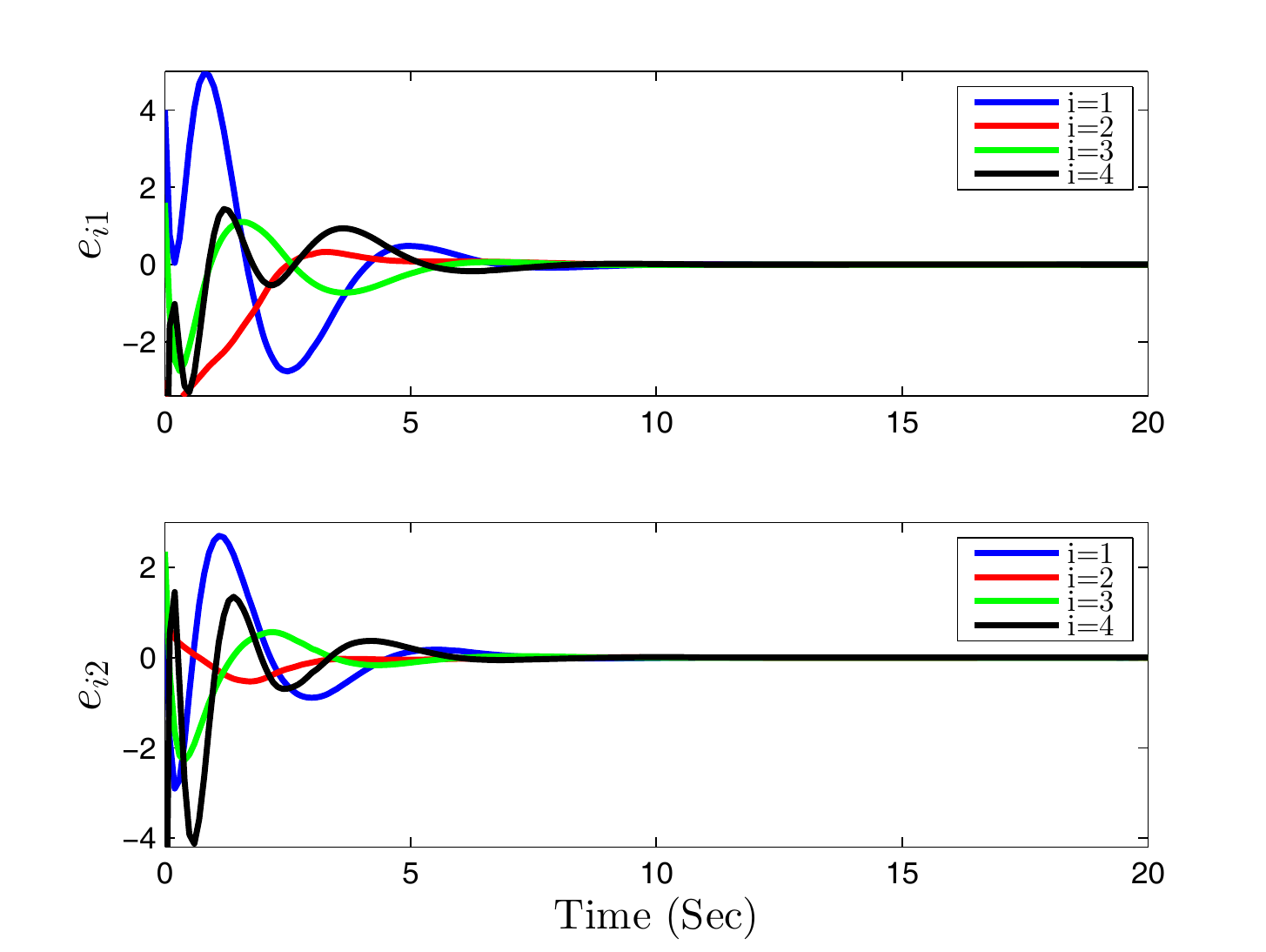}}
   \caption{Time history of output containment errors.}    \label{fig:6}
   \end{figure} 
 Figs.~\ref{fig:4} and ~\ref{fig:5} show the evaluation of learned optimal dynamic output feedback controls along with the state observer (\ref{eq:14}), and the adaptive observer (\ref{eq:35}) with adaptation laws (\ref{eq:36}) and (\ref{eq:37}) for the heterogeneous multi-agent systems (\ref{eq:98}). Fig.~\ref{fig:4} confirms that the followers move into the convex hull formed by the leaders' outputs. Note that, there exist directed paths from the leaders 5 and 6 to followers 1 and 2 and there is no directed path from the leader 7 to these followers, i.e., the leader 7 is disconnected from the followers 1 and 2. In this case, the convex hull is a line segment which is spanned by the outputs of these two leaders. Therefore, these two followers converge into the line segment which is spanned by the outputs of the leaders 5 and 6. Note also that, there exists a directed path from the leaders 5, 6, and 7 to the followers 3 and 4. Therefore, these two followers converge into the triangular area which is spanned by the outputs of the leaders 5, 6, and 7. The followers and the leaders' outputs along with the envelope of the leaders' outputs are shown in Fig.~\ref{fig:5}. It can be seen in Fig.~\ref{fig:5} that the followers' outputs move into the envelopes formed by the leaders' outputs and stay in them. The time history of output containment errors of followers is shown in Fig.~\ref{fig:6}. It can be observed that the containment errors decay to zero. Therefore, it can be seen that the followers' outputs converge to the convex hull formed by leaders' outputs, i.e., the containment control problem 1 is solved.
  
These results show that the introduced approach solves Problem 2 without requiring all followers to have knowledge of leaders' dynamics or states and based on only relative output measurements of followers with respect to their neighbors and the information which is broadcasted to them through the communication network.

\section{Conclusions}\label{sec7}

In this paper, an online optimal relative output-feedback based solution for the output containment control problem of linear multi-agent systems was presented. The followers were assumed heterogeneous in dynamics and dimensions. First, a distributed dynamic relative output-feedback control protocol was developed based on cooperative output regulation framework, which was provided the offline and non-optimal solution to the output containment control problem in hand. The proposed control protocol was composed of a feedback part and a feed-forward part. However, the feedback and feed-forward states were assumed to be unavailable to followers and were estimated using two distributed observers. A distributed adaptive observer was designed to estimate followers' states using only relative output measurements of followers with respect to their neighbors and the information which is broadcasted to them through the communication network. Another distributed observer was developed to estimate the convex hull of the leaders' states. To relax the restrictive assumption in existent work that each follower should have the knowledge of the leader's dynamics, an adaptive distributed observer was next designed to estimate both the leaders' dynamics and convex combination of the leaders' states. Optimality was explicitly imposed in finding the feedback and feed-forward control gains to not only assure convergence of followers' outputs to a convex combination of the leaders' outputs, but also optimize their transient output containment errors. To do this, augmented AREs were employed to solve the optimal output containment control problem in hand. An off-policy RL algorithm on an actor-critic structure was next developed to solve these AREs online in real time, based on using only relative output measurements of followers with respect to their neighbors and the information which is broadcasted to them through the communication network, and without requirement of knowing the complete knowledge of the leaders' dynamics by all followers. Finally, a simulation example verified the effectiveness of the proposed algorithm.



\begin{thebibliography}{1}

\bibitem{Jad2003}
Jadbabaie A, Jie Lin, Morse AS. Coordination of groups of mobile autonomous agents using nearest neighbor rules. IEEE Transactions on Automatic Control. 2003;48(6):988-1001. 

\bibitem{Ji2008}
Ji M, Ferrari-Trecate G, Egerstedt M, Buffa A. Containment Control in Mobile Networks. IEEE Transactions on Automatic Control. 2008;53(8):1972-1975.

\bibitem{Car2002}
Russell Carpenter J. Decentralized control of satellite formations. International Journal of Robust and Nonlinear Control. 2002;12(2-3):141-161.

\bibitem{Fax2004}
Fax JA, Murray RM. Information Flow and Cooperative Control of Vehicle Formations. IEEE Transactions on Automatic Control. 2004;49(9):1465-1476.

\bibitem{Tom1998}
Tomlin C, Pappas GJ, Sastry S. Conflict resolution for air traffic management: a study in multiagent hybrid systems. IEEE Transactions on Automatic Control. 1998;43(4):509-521. 

\bibitem{Ren2007}
Wei Ren, Beard RW, Atkins EM. Information consensus in multivehicle cooperative control. IEEE Control Systems Magazine. 2007;27(2):71-82.

\bibitem{Olfati-Saber2005}
Olfati-Saber R, Shamma JS. Consensus Filters for Sensor Networks and Distributed Sensor Fusion. In: Proceedings of the 44th IEEE Conference on Decision and Control. IEEE. 

\bibitem{Olfati-Saber2007}
Olfati-Saber R, Fax JA, Murray RM. Consensus and Cooperation in Networked Multi-Agent Systems. Proceedings of the IEEE. 2007;95(1):215-233.

\bibitem{Ren2005}
Wei Ren, Beard RW. Consensus seeking in multiagent systems under dynamically changing interaction topologies. IEEE Transactions on Automatic Control. 2005;50(5):655-661. 

\bibitem{Hong2006}
Hong Y, Hu J, Gao L. Tracking control for multi-agent consensus with an active leader and variable topology. Automatica. 2006;42(7):1177-1182.

\bibitem{Hong2008}
Hong Y, Chen G, Bushnell L. Distributed observers design for leader-following control of multi-agent networks. Automatica. 2008;44(3):846-850.

\bibitem{Lou2012}
Lou Y, Hong Y. Target containment control of multi-agent systems with random switching interconnection topologies. Automatica. 2012;48(5):879-885. 

\bibitem{Notarstefano2011}
Notarstefano G, Egerstedt M, Haque M. Containment in leader-follower networks with switching communication topologies. Automatica. 2011;47(5):1035-1040. 

\bibitem{Meng2010}
Meng Z, Ren W, You Z. Distributed finite-time attitude containment control for multiple rigid bodies. Automatica. 2010;46(12):2092-2099.

\bibitem{Cao2012a}
Cao Y, Ren W, Egerstedt M. Distributed containment control with multiple stationary or dynamic leaders in fixed and switching directed networks. Automatica. 2012;48(8):1586-1597. 

\bibitem{Liu2012a}
Liu H, Xie G, Wang L. Necessary and sufficient conditions for containment control of networked multi-agent systems. Automatica. 2012;48(7):1415-1422. 

\bibitem{Liu2012b}
Liu H, Xie G, Wang L. Containment of linear multi-agent systems under general interaction topologies. Systems and Control Letters. 2012;61(4):528-534.

\bibitem{Mei2012}
Mei J, Ren W, Ma G. Distributed containment control for Lagrangian networks with parametric uncertainties under a directed graph. Automatica. 2012;48(4):653-659. 

\bibitem{Li2012}
Li J, Ren W, Xu S. Distributed Containment Control with Multiple Dynamic Leaders for Double-Integrator Dynamics Using Only Position Measurements. IEEE Transactions on Automatic Control. 2012;57(6):1553-1559. doi:10.1109/tac.2011.2174680.

\bibitem{Yoo2013}
Yoo SJ. Distributed adaptive containment control of uncertain nonlinear multi-agent systems in strict-feedback form. Automatica. 2013;49(7):2145-2153. 

\bibitem{Li2013}
Li Z, Ren W, Liu X, Fu M. Distributed containment control of multi-agent systems with general linear dynamics in the presence of multiple leaders. International Journal of Robust and Nonlinear Control. 2011;23(5):534-547.

\bibitem{zheng2014}
Zheng Y, Wang L. Containment control of heterogeneous multi-agent systems. International Journal of Control. 2013;87(1):1-8.

\bibitem{Wang2014}
Wang Y, Cheng L, Hou Z-G, Tan M, Wang M. Containment control of multi-agent systems in a noisy communication environment. Automatica. 2014;50(7):1922-1928. 

\bibitem{Liu2015}
Liu H, Cheng L, Tan M, Hou Z-G. Containment control of continuous-time linear multi-agent systems with aperiodic sampling. Automatica. 2015;57:78-84. 

\bibitem{Haghshenas2015}
Haghshenas H, Badamchizadeh MA, Baradarannia M. Containment control of heterogeneous linear multi-agent systems. Automatica. 2015;54:210-216. 

\bibitem{Kan2015}
Kan Z, Klotz JR, Pasiliao EL Jr, Dixon WE. Containment control for a social network with state-dependent connectivity. Automatica. 2015;56:86-92.

\bibitem{Kan2016}
1. Kan Z, Shea JM, Dixon WE. Leader-follower containment control over directed random graphs. Automatica. 2016;66:56-62. doi:10.1016/j.automatica.2015.12.016.

\bibitem{Dimarogonas2009}
Dimarogonas DV, Tsiotras P, Kyriakopoulos KJ. Leader-follower cooperative attitude control of multiple rigid bodies. Systems and Control Letters. 2009;58(6):429-435. 

\bibitem{Mao2006}
Mahalik NP, ed. Chapter 13 - Sensor Networks and Configuration. Springer Berlin Heidelberg; 2007.

\bibitem{Shames2009}
Shames I, Fidan B, Anderson BDO. Minimization of the effect of noisy measurements on localization of multi-agent autonomous formations. Automatica. 2009;45(4):1058-1065. 

\bibitem{Wen2016}
Wen G, Zhao Y, Duan Z, Yu W, Chen G. Containment of Higher-Order Multi-Leader Multi-Agent Systems: A Dynamic Output Approach. IEEE Transactions on Automatic Control. 2016;61(4):1135-1140.

\bibitem{Zuo2017}
Zuo S, Song Y, Lewis FL, Davoudi A. Output Containment Control of Linear Heterogeneous Multi-Agent Systems Using Internal Model Principle. IEEE Transactions on Cybernetics. 2017;47(8):2099-2109.

\bibitem{Sutton1998}
Sutton RS, Barto AG. Reinforcement Learning-An Introduction. Cambridge, MA, USA: MIT Press, 1998.


\bibitem{Powell2007}
Powell WB. Approximate Dynamic Programming: Solving the Curses of Dimensionality. Hoboken, NJ, USA: Wiley, 2007.   






\bibitem{Cui2015}
Cui X, Zhang H, Luo Y, Jiang H. Adaptive dynamic programming for Hinf tracking design of uncertain nonlinear systems with disturbances and input constraints. International Journal of Adaptive Control and Signal Processing. 2017;31(11):1567-1583. 

\bibitem{Moghadam2017}
Moghadam R, Lewis FL. Output-feedback Hinf quadratic tracking control of linear systems using reinforcement learning. International Journal of Adaptive Control and Signal Processing. October 2017.

\bibitem{Vamvoudakis2017}
Vamvoudakis KG, Lewis FL, Dixon WE. Open-loop Stackelberg learning solution for hierarchical control problems. International Journal of Adaptive Control and Signal Processing. October 2017.

\bibitem{Yasini2014}
Yasini S, Karimpour A, Naghibi Sistani M-B, Modares H. Online concurrent reinforcement learning algorithm to solve two-player zero-sum games for partially unknown nonlinear continuous-time systems. International Journal of Adaptive Control and Signal Processing. 2014;29(4):473-493.

\bibitem{Modares2012}
Modares H, Lewis FL, Sistani M-BN. Online solution of nonquadratic two-player zero-sum games arising in the Hinf control of constrained input systems. International Journal of Adaptive Control and Signal Processing. 2012;28(3-5):232-254.

\bibitem{Zhang2011c}
Song R, Lewis FL, Wei Q, Zhang H. Off-Policy Actor-Critic Structure for Optimal Control of Unknown Systems With Disturbances. IEEE Transactions on Cybernetics. 2016;46(5):1041-1050.

\bibitem{DLiu2016}
Luo B, Liu D, Huang T, Wang D. Model-Free Optimal Tracking Control via Critic-Only Q-Learning. IEEE Transactions on Neural Networks and Learning Systems. 2016;27(10):2134-2144. 

\bibitem{DLiu2017}
Yang X, He H, Liu D. Event-Triggered Optimal Neuro-Controller Design With Reinforcement Learning for Unknown Nonlinear Systems. IEEE Transactions on Systems, Man, and Cybernetics: Systems. 2017:1-13.

\bibitem{DLiu2017b}
Luo B, Liu D, Wu H-N. Adaptive Constrained Optimal Control Design for Data-Based Nonlinear Discrete-Time Systems With Critic-Only Structure. IEEE Transactions on Neural Networks and Learning Systems. 2017:1-13.

\bibitem{Modares2016a}
Modares H, Nageshrao SP, Lopes GAD, Babuska R, Lewis FL. Optimal model-free output synchronization of heterogeneous systems using off-policy reinforcement learning. Automatica. 2016;71:334-341.

\bibitem{Yaghmaie2015}
Adib Yaghmaie F, Lewis FL, Su R. Output regulation of heterogeneous linear multi-agent systems with differential graphical game. International Journal of Robust and Nonlinear Control. 2015;26(10):2256-2278.


\bibitem{Tatari2016}
Tatari F, Naghibi-Sistani M-B, Vamvoudakis KG. Distributed learning algorithm for non-linear differential graphical games. Transactions of the Institute of Measurement and Control. 2016;39(2):173-182.

\bibitem{Tatari2017}
Tatari F, Naghibi-Sistani M-B, Vamvoudakis KG. Distributed optimal synchronization control of linear networked systems under unknown dynamics. In: 2017 American Control Conference (ACC). IEEE; 2017.

\bibitem{Mazouchi2018}
 Mazouchi M, Naghibi-Sistani MB, Sani SKH. A novel distributed optimal adaptive control algorithm for nonlinear multi-agent differential graphical games. IEEE/CAA Journal of Automatica Sinica. 2018;5(1):331-341. 

\bibitem{Qin2013}
Qin J, Yu C. Cluster consensus control of generic linear multi-agent systems under directed topology with acyclic partition. Automatica. 2013;49(9):2898-2905. 

\bibitem{Rockafellar1972}
Rockafellar RT. Convex analysis. New Jersey: Princeton University Press, 1972.

\bibitem{Huang2004}
Huang J. Nonlinear Output Regulation: Theory and Applications. SIAM: Philadelphia, 2004.


\bibitem{Huang2016}
Huang J. Chapter Fourteen - Certainty Equivalence, Separation Principle, and Cooperative Output Regulation of Multi-agent Systems by the Distributed Observer Approach. In: Control of Complex Systems. Elsevier; 2016:421-449.

\bibitem{Gadewadikar07}
Gadewadikar J, Lewis FL, Xie L, Kucera V, Abu-Khalaf M. 
Parameterization of all stabilizing Hinf static state-feedback gains: Application to output-feedback design. Automatica, 2007; 43(9):1597-1604.
 
\bibitem{Zhang2011}
 Hongwei Zhang, Lewis FL, Das A. Optimal Design for Synchronization of Cooperative Systems: State Feedback, Observer and Output Feedback. IEEE Transactions on Automatic Control. 2011;56(8):1948-1952.

\bibitem{Wu2016}
Wu J, Ugrinovskii V, Allgower F. Observer-based synchronization with relative measurements and unknown neighbour models. In: 2016 Australian Control Conference (AuCC). IEEE; 2016.

\bibitem{Cai2017}
 Cai H, Lewis FL, Hu G, Huang J. The adaptive distributed observer approach to the cooperative output regulation of linear multi-agent systems. Automatica. 2017;75:299-305.

\bibitem{LewisBook}
Lewis FL, Syrmos VL. Optimal Control. John Wiley and Sons: 1995.

\bibitem{Modares2015}
Modares H, Lewis FL, Jiang Z-P. $ {H}_{ {\infty }}$ Tracking Control of Completely Unknown Continuous-Time Systems via Off-Policy Reinforcement Learning. IEEE Transactions on Neural Networks and Learning Systems. 2015;26(10):2550-2562. 

\end{thebibliography}





\end{document}